\documentclass[11pt]{article}

\usepackage{amsmath,amssymb,amsthm}

\newcommand{\R}{\mathbb{R}}
\newcommand{\C}{\mathbb{C}}

\newcommand{\K}{{\cal K}}
\newcommand{\F}{{\cal F}}

\newcommand{\pr}{{\rm pr}}

\newcommand{\D}{{\cal D}}
\newcommand{\Abar}{\overline{\cal A}}

\newcommand{\h}{{\cal H}}
\newcommand{\scal}[2]{\langle #1| #2\rangle}

\newcommand{\eps}{\epsilon}

\newcommand{\Cyl}{{\rm Cyl}}

\newcommand{\ot}{\otimes}
\newcommand{\bld}[1]{\boldsymbol{#1}}
\newcommand{\tl}[1]{\tilde{#1}}

\newcommand{\la}{\lambda}
\DeclareMathOperator{\tr}{{\rm tr}}
\DeclareMathOperator{\spn}{{\rm span}}
\newcounter{mnotecount}[section]

\newtheorem{thr}{Theorem}
\newtheorem{lm}[thr]{Lemma}
\newtheorem{df}[thr]{Definition}
\newtheorem{cor}[thr]{Corollary}
\newtheorem{pro}[thr]{Proposition}

\numberwithin{equation}{section}
\numberwithin{thr}{section}

\begin{document}

\title{Construction of spaces of kinematic quantum states for field theories via projective techniques\footnote{This is an author-created, un-copyedited version of an article accepted for publication in Classical and Quantum Gravity. IOP Publishing Ltd is not responsible for any errors or omissions in this version of the manuscript or any version derived from it. The definitive publisher authenticated version is available online at http://dx.doi.org/10.1088/0264-9381/30/19/195003.}}
\author{ Andrzej Oko{\l}\'ow}
\date{September 10, 2013}

\maketitle
\begin{center}
{\it  1. Institute of Theoretical Physics, Warsaw University\\ ul. Ho\.{z}a 69, 00-681 Warsaw, Poland\smallskip\\
2. Department of Physics and Astronomy, Louisiana State University,\\
Baton Rouge, LA 70803, USA\smallskip\\
oko@fuw.edu.pl}
\end{center}
\medskip

\begin{abstract}
We present a method of constructing a space of quantum states for a field theory: given phase space of a theory, we define a family of physical systems each possessing a finite number of degrees of freedom, next we define a space of quantum states for each finite system, finally using projective techniques we organize all these spaces into a space of quantum states which corresponds to the original phase space. This construction is kinematic in this sense that it bases merely on the structure of the phase space of a theory and does not take into account possible constraints on the space. The construction is a generalization of a construction by Kijowski---the latter one is limited to theories of linear phase spaces, while the former one is free of this limitation. The method presented in this paper enables to construct a space of quantum states for the Teleparallel Equivalent of General Relativity.
\end{abstract}

%***************************************************
\section{Introduction \label{intro}}
%***************************************************

In this paper we are concerned with a construction of a space of quantum states for a field theory which could serve as an element of canonical quantization (possibly, a modification of canonical quantization) of the theory. More precisely, we would like to construct a space of quantum states associated with the phase space of a field theory neglecting at this stage possible constraints on the phase space---here we adopt the Dirac's approach to canonical quantization of constrained systems: ``first quantize, then solve the constraints''. In this sense the space of quantum states we are going to construct is ``kinematic''. 

The main motivation in addressing the issue of constructing such spaces is our wish to quantize canonically the Teleparallel Equivalent of General Relativity (TEGR). This theory is background independent (see \cite{mal-rev} for the newest review of TEGR) and therefore it is quite natural to try to build a space of quantum states for it in a background independent manner. However, the only well developed background independent methods of constructing spaces of quantum states, that is, Loop Quantum Gravity (LQG) methods (see e.g. \cite{rev,rev-1} for reviews of LQG), for a reason described below do not seem to be useful in the case of TEGR. This fact motivated us to develop  an other method of constructing such spaces.

The task of constructing a space of kinematic quantum states for a field theory is not very easy because an unconstrained phase space of a field theory is an ``infinite-dimensional'' space. To overcome difficulties connected with this property of the phase space one can proceed in the following way: first one reduces the phase space to a finite dimensional space representing a physical system of a finite number of degrees of freedom. Once one has defined sufficiently many finite systems in this sense that they altogether contain all relevant information encoded in the original phase space one defines a space of quantum states (e.g. a Hilbert space) for each finite system. The last step of this procedure consists in building a ``large'' space of quantum states from ``small'' spaces of quantum states, that is, from spaces of quantum states associated with all the finite systems. Assuming that the finite systems are labeled by a {\em directed set} one can try to equip the family of ``small'' spaces with the structure of either $(i)$ an inductive family or $(ii)$ a projective family and then define the ``large'' space as the corresponding inductive or projective limit. 

It turns out that both inductive and projective techniques can be quite successful in particular cases.  An example of a space of quantum states obtained by an application of inductive methods is the kinematic Hilbert space of LQG \cite{al-hoop}\footnote{In the original paper \cite{al-hoop} presenting the construction of the Hilbert space inductive techniques are not used explicitely, but it can be shown that this Hilbert space is the inductive limit of a family of Hilbert spaces associated with some finite systems---see e.g. \cite{oko-ncomp} for some details.}, a space of quantum states built by means of projective methods can be found in a paper \cite{kpt} by Kijowski (see also \cite{kpt-qed}). However, both methods have their limitations: the inductive techniques applied in LQG do not work well if the configuration space (that is, the space of ``positions'') of each finite system is {\em non-compact} \cite{oko-ncomp}; the Kijowski's projective techniques work only for field theories of {\em linear} phase spaces i.e. of phase spaces possessing the structure of a linear (vector) space\footnote{The Kijowski's projective techniques need a linear structure to be defined on the phase space of a theory. However, they do not require this linear structure to be unique or ``canonical''. Therefore the techniques are applicable if e.g. a phase space is an affine space \cite{kpt-qed}.}. Regarding applicability of both methods to TEGR: the LQG methods do not seem to be applicable to TEGR since in this case it is rather difficult to define naturally finite systems of compact configuration spaces. On the other hand, the phase space of TEGR is non-linear (see e.g. \cite{oko-tegr}) and therefore the Kijowski's method cannot be applied either.

In this paper we present a generalization of the Kijowski's construction which is not limited by linearity of a phase space. This generalization is, roughly speaking, based on two assumptions: $(i)$ the phase space is a Cartesian product $P\times Q$ of a space $P$ of momenta and a space $Q$ of ``positions'', where in general $Q$ {\em is not} a linear space $(ii)$ there exists a family of functions on $Q$ of special properties which define an embedding of $Q$ into a linear space. Let us emphasize that the resulting space of quantum states is not a Hilbert space but it is rather a convex set of quantum states which can be identified with algebraic states (i.e. linear positive normed functionals) on a $C^*$-algebra which on the other hand can be interpreted as an algebra of some quantum observables \cite{kpt}.  

The generalized method can be applied both to background independent and background dependent theories. Since the method encompasses the original Kijowski's construction it is justified to say that the generalized method provides a space of kinematic quantum states for a scalar field theory \cite{kpt}. Moreover, as shown in the present paper the method can be successfully applied to a simple background independent field theory introduced in \cite{oko-ncomp} called here {\em degenerate Pleba\'nski gravity (DPG)}. Finally, the method provides a space of quantum states for TEGR---a detailed description of this construction will be presented in forthcoming papers \cite{q-suit,ham-nv,q-stat}. In other words, in the case of TEGR the method  allows to carry out the first step of the Dirac's procedure, that is, ``first quantize'' (as it can be expected, the second step, i.e. ``solve constraints'', is in this case much harder than the first one and by now we are still far away from a solution).

Let us mention also, that there is a by-product of the generalized method. Although the method uses projective techniques it turns out that for each theory for which the method works well one can additionally construct another space of quantum states using inductive techniques. This another space is a Hilbert space constructed from some almost periodic functions.

The paper is organized as follows: Section 2 contains preliminaries, in Section 3 being the main part of the paper we present detailed assumptions underlying the  generalized construction of spaces of quantum states, carry out the construction and finally present some facts which may be helpful while applying the generalized method in practice. Section 4 describes the by-product of the construction, that is, the Hilbert space mentioned above. In Section 5 we present an example of an application of the generalized method---we construct a space of quantum states for DPG. Finally, Section 6 contains a discussion of the results.

%***************************************************
\section{Preliminaries \label{pre}}
%***************************************************

The construction of spaces of quantum states we are going to describe in this paper will be obviously patterned on the original one by Kijowski \cite{kpt}---a detailed comparison of the construction presented here and that of \cite{kpt} can be found in Section \ref{disc-comp}. Additionally we will apply to the construction some notions like e.g. cylindrical functions or some linear operators on them taken from LQG (see e.g. \cite{acz,cq-diff,rev,rev-1} and references therein).  

Consider a Hamiltonian formulation of a field theory. Assume that the phase space of the formulation is a Cartesian product $P\times Q$, where $P$ is a space of momenta and $Q$ is a space of ``positions'' which will be called here {\em a Hamiltonian configuration space}. A point $p\in P$ as well as a point $q\in Q$ is a (collection of) field(s) on a manifold $\Sigma$ which can be thought of as a space-like slice of a spacetime.      

Let $\K$ be a collection of real functions on $Q$. The functions in $\K$ will be called {\em (configurational) elementary  degrees of freedom} if they separate points in $Q$. A configurational elementary d.o.f. will be usually denoted by $\kappa$ possibly with some indeces. 

Similarly, let ${\cal F}$ be a collection of real functions on $P$. The functions in $\F$ will be called {\em (momentum) elementary  degrees of freedom} if they separate points in $P$. A momentum elementary d.o.f. will be usually denoted by $\varphi$ possibly with some indeces.

Let $K=\{\kappa_{1},\ldots,\kappa_N\}\subset{\cal K}$ be a finite set of elementary d.o.f.. We say that $q\in Q$ is $K$-related to $q'\in Q$, 
\[
q\sim_{K} q',
\]
if for every $\kappa_i\in{K}$      
\[
\kappa_i(q)=\kappa_i(q').
\]
Clearly, the relation $\sim_{K}$ is an equivalence one. Therefore it is meaningful to consider the following quotient space
\begin{equation}
Q_{K}:=Q/\sim_{K}.
\label{quot}
\end{equation}
Note now that there exist an {\em injective} map\footnote{A set $K$ is unordered, thus to define the map $\tilde{K}$ one has to order elements of $K$. However, every choice of the ordering is equally well suited for our purposes and nothing essential depends on the choice. Therefore we will neglect this subtlety throughout the paper.} from $Q_{K}$ into $\R^N$:
\begin{equation}
Q_{K}\ni[q]\mapsto\tilde{K}([q]):=(\kappa_{1}(q),\ldots,\kappa_{N}(q))\in\R^N,
\label{k-inj}
\end{equation}
where $N$ is the number of elementary d.o.f. constituting $K$ and $[q]$ denotes the equivalence class of $q$ defined by the relation $\sim_{K}$.       

We will say that elementary d.o.f. in ${K}=\{\kappa_{1},\ldots,\kappa_{N}\}$ are {\em independent} if the image of $\tilde{K}$ is an $N$-dimensional submanifold of $\R^N$. Then the map $\tilde{K}$ is a bijection onto its image and therefore it can be used to pull-back the topology and the differential structure from the image onto $Q_{K}$. In this case functions $(x_1,\ldots,x_N)$ defined as  
\begin{equation}
Q_K\ni[q]\mapsto x_i([q]):=\kappa_i(q)\in\R
\label{lin-coor}
\end{equation}
constitute a {\em global} coordinate frame on $Q_K$. Note that 
\[
(x_1([q]),\ldots,x_N([q]))=\tilde{K}([q]).
\] 

It may happen that for distinct sets of independent d.o.f. ${K}$ and ${K}'$ the spaces $Q_{K}$ and $Q_{{K}'}$ coincide, i.e. for every $q\in Q$ 
\[
[q]=[q]',
\]
where $[q]'$ is the equivalence class defined by the relation $\sim_{K'}$. We assume that then the topologies and differential structures defined on $Q_{K}=Q_{{K}'}$ by $\tilde{K}$ and $\tilde{K}'$ coincide as well---in other words we assume that  $\tilde{K}'\circ\tilde{K}^{-1}$ is a diffeomorphism from the image of  $\tilde{K}$ onto the image of $\tilde{K}'$.            

\begin{df}
A quotient space $Q_{K}$ will be called  a reduced configuration space if the elementary  d.o.f. constituting $K$ are independent. 
\end{df}
Since now the symbol $Q_{K}$ will denote a reduced configuration space defined by the set $K$ of independent d.o.f..

Denote by $\pr_{K}$ the canonical projection from $Q$ onto $Q_{K}$:
\begin{equation}
Q\ni q\mapsto\pr_{K}(q)=[q]\in Q_{K}.
\label{pr-K}
\end{equation}

\begin{df}
We say that a function $\Psi:Q\to\C$ is a cylindrical function compatible with the set ${K}$ of independent d.o.f. if 
\begin{equation}
\Psi=\pr^*_{K}\,\psi
\label{Psi-cyl}
\end{equation}
for some smooth function $\psi:Q_{K}\to\C$. 
\end{df}

Suppose now that $Q_{K}=Q_{{K}'}$. Then $\Psi$ is a cylindrical function compatible with $K$ if and only if  $\Psi$ is a cylindrical function compatible with ${K}'$.  

Denote by $\Cyl$ a complex linear space of functions on $Q$ spanned by all the cylindrical functions on $Q$. We assume that each elementary momentum d.o.f. $\varphi$ defines a linear operator $\hat{\varphi}$ on $\Cyl$ via the Poisson bracket on the phases space $P\times Q$
\begin{equation}
\Cyl\ni\Psi\mapsto\hat{\varphi}\Psi:=\{\varphi,\Psi\}\in\Cyl
\label{hat-varphi}
\end{equation}
or, if necessary, a suitable regularization of the Poisson bracket at the r.h.s. of this formula{\footnote{Arguments of a Poisson bracket are usually functionals on the phase space defined via three-dimensional integrals over $\Sigma$. But if $q$ is a one-form on $\Sigma$ and $p$ a two-form then it is natural to define elementary d.o.f. via integrals of the forms over submanifolds of $\Sigma$ of appropriate dimensions. In such a case the Poisson bracket $\{\varphi,\Psi\}$ needs a regularization. An example of an operator $\hat{\varphi}$ defined via a regularization of such a Poisson bracket is the flux operator of LQG \cite{acz,rev-1}---the regularization consists in smearing components of $p$ with a three-dimensional smearing field which tends to a distribution supported on a two-dimensional surface. Importantly, the regularization is compatible with the diffeomorphism invariance of LQG. In fact, operators $\{\hat{\varphi}\}$ used to construct a space of quantum states for DPG in Section \ref{DPG} and for TEGR in \cite{q-stat} are such flux operators.}}. 
  
We will denote by $\hat{\cal F}$ a real linear space of operators on $\Cyl$ spanned by the operators \eqref{hat-varphi}:
\[
\hat{\F}:=\spn_\R\{\ \hat{\varphi}\ | \ \varphi\in\F\ \}.
\]    

%***************************************************
\section{Construction of the space of quantum states}
%***************************************************

In this section we will construct a space $\D$ of quantum states for a theory of the phase space $P\times Q$.

%***************************************************
\subsection{Outline of the construction \label{outline}}
%*************************************************** 

Before we will describe the construction let us denote by $\hat{\mathbf{F}}$ the set of all finite dimensional subspaces of $\hat{\cal F}$ and by $\mathbf{K}$ the set of all finite sets of {\em independent} elementary d.o.f.. An element of $\hat{\mathbf{F}}$ will be usually denoted by $\hat{F}$, and an element of $\mathbf{K}$ by $K$.   

The start point of the construction will be a suitably chosen subset $\Lambda$ of $\hat{\mathbf{F}}\times\mathbf{K}$ equipped with a binary relation $\geq$ defining on it a structure of a directed set. The construction of the space of quantum states will consists of the following steps:
\begin{enumerate}
\item  we will show that it follows from properties of $(\Lambda,\geq)$ that if a pair of its elements satisfy $(\hat{F}',{K}')\geq(\hat{F},{K})$ then there exists a projection $\pr_{{K}{K}'}$ from $Q_{K'}$ onto $Q_K$ and for every triplet $(\hat{F}'',{K}'')\geq(\hat{F}',{K}')\geq(\hat{F},{K})$ the corresponding projections satisfy
\[
\pr_{KK''}=\pr_{KK'}\circ\pr_{K'K''}.
\]
\item next, we will show that, given a pair $\lambda'\equiv(\hat{F}',{K}')\geq\lambda\equiv(\hat{F},{K})$, there exists a distinguished injection $\omega_{\lambda'\lambda}:Q_K\to Q_{K'}$ such that 
\[
\pr_{KK'}\circ\omega_{\lambda'\lambda}={\rm id} 
\]
and
\[
\omega_{\lambda''\lambda}=\omega_{\lambda''\lambda'}\circ\omega_{\lambda'\lambda}
\]
for every triplet $\lambda''\geq\lambda'\geq\lambda$ of elements of $\Lambda$.
\item for every $\lambda\equiv(\hat{F},K)\in\Lambda$ we will define a Hilbert space $\h_\lambda$ as a space of functions on $Q_K$ square integrable with respect to a measure on $Q_K$. 
\item we will show that, given $\lambda'\equiv(\hat{F}',K')\geq\lambda\equiv(\hat{F},K)$, the projection $\pr_{KK'}$ and the injection $\omega_{\lambda'\lambda}$ define a projection $\pi_{\lambda\lambda'}$ from a space $\D_{\lambda'}$ of all density operators (i.e. positive operators of trace equal $1$) on $\h_{\lambda'}$ onto a space $\D_\lambda$ defined analogously. For every triplet $\lambda''\geq\lambda'\geq\lambda$ the corresponding projections will satisfy   
\[
\pi_{\lambda\lambda''}=\pi_{\lambda\lambda'}\circ\pi_{\lambda'\lambda''},
\]
which means that the set $\{\D_\lambda,\pi_{\lambda\lambda'}\}_{\lambda\in\Lambda}$ will be a projective family;

\item finally, the space of quantum states $\D$ will be defined as the projective limit of $\{\D_\lambda,\pi_{\lambda\lambda'}\}_{\lambda\in\Lambda}$. 
\end{enumerate}

%***************************************************
\subsection{Additional assumptions \label{ad-as}}
%***************************************************

To ensure that the construction outlined in the previous subsection will work we have to impose on the directed set $(\Lambda,\geq)$ some {\bf Assumptions}:
\begin{enumerate}
\item 
\begin{enumerate}
\item for each finite set $K_0$ of configurational elementary d.o.f.  there exists $(\hat{F},K)\in\Lambda$ such that each $\kappa\in K_0$ is a cylindrical function compatible with $K$; \label{k-Lambda}
\item for each finite set $F_0$ of momentum elementary d.o.f. there exists $(\hat{F},K)\in\Lambda$ such that $\hat{\varphi}\in\hat{F}$ for every $\varphi\in F_0$; \label{f-Lambda}
\end{enumerate}
\item \label{RN} 
if $(\hat{F},K)\in\Lambda$ then the image of the map $\tilde{K}$ given by \eqref{k-inj} is $\R^N$ (where $N$ is the number of elements of $K$)---in other words, $\tilde{K}$ is a bijection and consequently  
\[
Q_K\cong\R^N.
\] 
\item 
if $(\hat{F},K)\in\Lambda$, then 
\begin{enumerate}
\item for every $\hat{\varphi}\in \hat{\F}$ and for every cylindrical function $\Psi=\pr_K^*\psi$ compatible with $K=\{\kappa_1,\ldots,\kappa_N\}$ 
\[
\hat{\varphi}\Psi=\sum_{i=1}^N\Big(\pr^*_K\partial_{x_i}\psi\Big)\hat{\varphi}\kappa_i,
\]   
where $\{\partial_{x_i}\}$ are vector fields on $Q_K$ given by the global coordinate frame \eqref{lin-coor}; \label{comp-f} 
\item for every $\hat{\varphi}\in \hat{\F}$ and for every $\kappa\in K$ the cylindrical function $\hat{\varphi}\kappa$ is a real {\em constant} function on $Q$; \label{const}
\end{enumerate}
\item if $(\hat{F},K)\in\Lambda$ and $K=\{\kappa_{1},\ldots,\kappa_{N}\}$ then $\dim\hat{F}=N$; moreover, if $(\hat{\varphi}_1,\ldots,\hat{\varphi}_N)$ is a basis of $\hat{F}$ then an $N\times N$ matrix $G=(G_{ji})$ of components
\[
G_{ji}:=\hat{\varphi}_j\kappa_i
\]    
is {\em non-degenerate}. \label{non-deg}
\item  if $(\hat{F},K'),(\hat{F},K)\in\Lambda$ and $Q_{K'}=Q_{K}$ then  $(\hat{F},K')\geq(\hat{F},K)$; \label{Q'=Q} 
\item if $(\hat{F}',K')\geq(\hat{F},K)$ then 
\begin{enumerate}
\item each d.o.f. $K$ is {\em a linear combination} of d.o.f. in $K'$; \label{lin-comb}
\item $\hat{F}\subset\hat{F}'$. \label{FF'}
\end{enumerate} 
\end{enumerate} 

%***************************************************
\subsection{Physical motivation for the construction}
%***************************************************

The physical motivation for this construction is exactly the same as that presented in \cite{kpt}. According to it each elementary d.o.f. represents a measuring device by means of which one can extract some information about a point $(p,q)$ in the phase space. Consequently, a pair $(\hat{F},K)\in\Lambda$ represents a finite collection of such devices: $\{\varphi_1,\ldots,\varphi_N\}$ generating a basis $\{\hat{\varphi}_1,\ldots,\hat{\varphi}_N\}$ of $\hat{F}$ and $\{\kappa_1,\ldots,\kappa_N\}$ constituting $K$ and defining a reduced configuration space $Q_K$. Thus, given a point $(p,q)$ in the phase space, a pair $(\hat{F},K)$ provides us with a partial information about the point. Therefore a pair $(\hat{F},K)$ can be called a {\em reduced physical system} where the word ``reduced'' is used with respect to the complete system represented by the phase space $P\times Q$: uncountably many d.o.f. of the phase space are reduced to a finite number of d.o.f. defining $(\hat{F},K)$. Assumption \ref{non-deg} in cases when operators \eqref{hat-varphi} are given by a genuine Poisson bracket amounts to the non-degeneracy of the bracket on the corresponding reduced system or, equivalently, to the existence of a symplectic structure on the phase space of the reduced system.

The relation $\geq$ between points of $\Lambda$, that is, between reduced systems is meant to be a relation {\em system--its subsystem}: if $\lambda'\equiv(\hat{F}',K')\geq \lambda\equiv(\hat{F},K)$ then by virtue of Assumption \ref{FF'} the space $\hat{F}'$ representing momentum d.o.f. of the system $\la'$ contains the space $\hat{F}$ representing momentum d.o.f. of the system $\la$; on the other hand, by virtue of Assumption \ref{lin-comb} configurational elementary d.o.f. constituting $K'$ contain complete information about configurational elementary d.o.f. constituting $K$. Thus it is fully justified to call $\la$ a {\em subsystem} of $\la'$.

The Hilbert space $\h_\lambda$ introduced in Step 3 of the construction outline describes pure quantum states of the system $\la$ while the space $\D_\lambda$ introduced in Step 4 represents mixed states of the system. The task of the projection $\pi_{\la\la'}:D_{\la'}\to \D_\la$ is to reduce quantum d.o.f. of a system $\la'$ to quantum d.o.f. of its subsystem $\la$---as one can expect, such a projection will be defined by means of an appropriate partial trace. Finally, the projections $\{\pi_{\la\la'}\}$ are used to ``glue'' all spaces $\{\D_\la\}$ into a large space $\D$ representing all quantum states of the original theory.

%***************************************************
\subsection{The construction}
%***************************************************

We assume that for every set $K\in\mathbf{K}$ considered in this section there exists $\hat{F}\in\hat{\mathbf{F}}$ such that $(\hat{F},K)\in\Lambda$---this requirement will allow us to use all Assumptions listed in Section \ref{ad-as}. It does not mean, however, that the requirement is a necessary condition for every fact proven below.

%***************************************************
\subsubsection{Step 1}
%***************************************************

Let us recall that the goal of this step is to prove the existence of projections $p_{KK'}:Q_{K'}\to Q_K$ and to describe their properties.

It follows from Assumption \ref{RN} that on every reduced configuration space $Q_K$ under considerations the corresponding map $\tilde{K}$ given by \eqref{k-inj} defines a structure of a linear space which is the linear structure of $\R^N$ pulled back on $Q_K$. Then the coordinate frame \eqref{lin-coor} is linear.

\begin{lm}
Suppose that independent d.o.f in $K=\{\kappa_{1},\ldots,\kappa_{N}\}$ are linear combinations of  independent d.o.f. in $K'=\{\kappa'_1,\ldots,\kappa'_{N'}\}$. Then the linear combinations generate a linear projection $\pr_{KK'}:Q_{K'}\to Q_K$ such that
\[
\pr_K=\pr_{KK'}\circ\pr_{K'}.
\]
Moreover, if $N'=N$ then $Q_{K'}=Q_{K}$ and $\pr_{KK'}={\rm id}$.   
\label{lm-pr} 
\end{lm}

\begin{proof}
The assumption in the lemma  means that there exists a constant matrix $B=(B_i{}^{j})$ where $i=1,2,\ldots,N$ and $j=1,2,\ldots,N'$ such that 
\[
\kappa_i=B_i{}^{j}\kappa'_{j},
\]
where a summation over $j$ is assumed. Consequently,  
\begin{equation}
\kappa_i(q)=B_i{}^{j}\kappa'_{j}(q)
\label{kqk'q}
\end{equation}
for every $q\in Q$ and for every $i=1,\ldots,N$. This formula can be expressed in terms of the maps $\tilde{K}$ and $\tilde{K}'$ defined by \eqref{k-inj}:
\begin{equation}
\tilde{K}([q])=B\tilde{K}'([q]'),
\label{tk-btk'}
\end{equation}
where $[q]\in Q_K$ and $[q]'\in Q_{K'}$ and the symbols at the r.h.s. denote the standard  action of the matrix $B$ on the vector $\tilde{K}'([q]')$. Because $\tilde{K}$ is injective 
\begin{equation}
[q]=\tilde{K}^{-1}(B\tilde{K}'([q]'))
\label{qKqK'}
\end{equation}
for every $q\in Q$. This equation defines a map 
\begin{equation}
\pr_{KK'}:=\tilde{K}^{-1}(B\tilde{K}')
\label{pr-KK}
\end{equation}
from $Q_{K'}$ into $Q_K$. In other words,
\begin{equation}
\pr_{KK'}([q]')=[q].
\label{prq'q}
\end{equation}
Since $\tilde{K}'$ and $\tilde{K}$ define linear structures of, respectively, $Q_{K'}$ and $Q_K$ the map $\pr_{KK'}$ is {\em linear}. Moreover, because     \eqref{qKqK'} holds for every $q\in Q$ the map $\pr_{KK'}$ is a {\em projection}. 

Using maps $\pr_K$ and $\pr_{K'}$ defined by \eqref{pr-K} and the projection \eqref{pr-KK} we rewrite \eqref{qKqK'} obtaining
\begin{equation}
\pr_K=\pr_{KK'}\circ\pr_{K'}.
\label{ppKKp}
\end{equation}

Note that because $\pr_{KK'}$ is a linear projection then $N'=\dim Q_{K'}\geq\dim Q_{K}=N$. Suppose that $N=N'$. Then $\pr_{KK'}$ is a linear {\em isomorphism} and consequently the matrix $B$ in \eqref{kqk'q} is invertible. Hence for every $q\in Q$ and for every $j=1,\ldots,N'$  
\[
\kappa'_j(q)=B^{-1}_j{}^i\kappa_i(q).
\]
It follows from this equation and Equation \eqref{kqk'q} that $q\sim_K q'$ if and only if $q\sim_{K'}q'$ and therefore $Q_K=Q_{K'}$. Then $[q]=[q]'$. Setting this to \eqref{qKqK'} we obtain  
\[
[q]=\tilde{K}^{-1}(B\tilde{K}'([q])),
\]   
which means that $\pr_{KK'}={\rm id}$.
\end{proof}

Assumptions \ref{Q'=Q} and \ref{lin-comb}  together with Lemma \ref{lm-pr} guarantee that in the case of sets $K$ and $K'$ of independent d.o.f. such that $(\hat{F},K'),(\hat{F},K)\in\Lambda$ for some $\hat{F}\in\hat{\mathbf{F}}$ and $Q_K=Q_{K'}$ the projection $\pr_{KK'}$ exists and is the identity map on $Q_K=Q_{K'}$. Moreover, then the matrix $B$ is invertible and  Equation \eqref{tk-btk'} reads 
\[
\tilde{K}([q])=B\tilde{K}'([q]),
\]
which allows us to conclude that the linear structure defined on $Q_{K}=Q_{K'}$ by $\tilde{K}$ coincides with one defined by $\tilde{K}'$.

Let $K''=\{\kappa''_{1},\ldots,\kappa''_{N''}\}$, $K'=\{\kappa'_{1},\ldots,\kappa'_{N'}\}$ and $K=\{\kappa_{1},\ldots,\kappa_{N}\}$ be sets of independent d.o.f.. Suppose that $(i)$ every function in $K$ is a linear combination of the functions in $K'$, $(ii)$  every function in $K'$ is a linear combination of the functions in $K''$ and $(iii)$ every function in $K$ is a linear combination of the functions in $K''$. Thus there exist constant matrices $B=(B_i{}^{j})$, $C=(C_{j}{}^{l})$ and $D=(D_i{}^{l})$ ($i=1,\ldots, N$, $j=1,\ldots,N'$ and $l=1,\ldots,N''$)  such that
\begin{align*}
&\kappa_i=B_{i}{}^{j}\kappa'_{j},&&\kappa'_{j}=C_{j}{}^{l}\kappa''_{l},&&\kappa_i=D_{i}{}^{l}\kappa''_{l}.
\end{align*} 
Consequently,
\[
\kappa_i=D_{i}{}^{l}\kappa''_{l}=B_{i}{}^{j}C_{j}{}^{l}\kappa''_{l}
\]
or, equivalently,
\[
\tilde{K}([q])=D\tilde{K}''([q]'')=BC\tilde{K}''([q]''),
\]
for every $q\in Q$, where $[q]\in Q_K$ and $[q]''\in Q_{K''}$. This means that
\[
\tilde{K}^{-1}(D\tilde{K}'')=\tilde{K}^{-1}(BC\tilde{K}'')=\tilde{K}^{-1}\Big(B\tilde{K}'[\tilde{K}^{\prime-1}(C\tilde{K}'')]\Big).
\]
Using \eqref{pr-KK} we rewrite the formula above obtaining
\begin{equation}
\pr_{KK''}=\pr_{KK'}\circ\pr_{K'K''}.
\label{pr-KKK}
\end{equation}   

Note that so far in this step we used solely relations between sets of configurational elementary d.o.f. without any application operators in $\hat{\F}$. 

By virtue of Assumption \ref{lin-comb} 
\begin{enumerate}
\item for every pair $(\hat{F}',K')\geq (\hat{F},K)$ of elements of $\Lambda$ there exists a {\em linear projection} $\pr_{KK'}:Q_{K'}\to Q_{K}$ defined by \eqref{pr-KK};
\item for every triplet $(\hat{F}'',{K}'')\geq(\hat{F}',{K}')\geq(\hat{F},{K})$ of elements of $\Lambda$ the corresponding projections $\pr_{KK''}$, $\pr_{KK'}$ and $\pr_{K'K''}$ satisfy \eqref{pr-KKK}.  
\end{enumerate}

Regarding the physical meaning of the projection $\pr_{KK'}$ one can say that the projection reduces configurational d.o.f. of the space $Q_{K'}$ to those of $Q_K$.  

%***************************************************
\subsubsection{Step 2}
%***************************************************

In this step we will show that if $(\hat{F}',K')\geq (\hat{F},K)$ then there exists a natural injection $\omega_{K'K}:Q_K\to Q_{K'}$.   

Let us start by showing that for $K=\{\kappa_1,\ldots,\kappa_N\}\in\mathbf{K}$ by virtue of Assumptions \ref{comp-f} and \ref{const} there exists a natural linear map from $\hat{\cal F}$ into $Q_K$. Consider a cylindrical function $\Psi$ compatible with $K$. If $\hat{\varphi}\in\hat{\cal F}$ then by virtue of Assumption \ref{const} each cylindrical function $\hat{\varphi}\kappa_i$ is constant and therefore since now it will be treated as a {\em real number}. This conclusion together with Assumption \ref{comp-f} gives us
\begin{equation}
\hat{\varphi}\Psi=\pr_K^*\sum_i(\hat{\varphi}\kappa_i)\partial_{x_i}\psi,
\label{varphi-Psi}
\end{equation}
where $\{\partial_{x_i}\}$ are vector fields on $Q_K$ given by the linear coordinate frame $(x_i)$ on $Q_K$ defined by \eqref{lin-coor}. Since the coordinates $(x_i)$ are linear and $\{\hat{\varphi}\kappa_i\}$ are numbers the vector field $\sum_i(\hat{\varphi}\kappa_i)\partial_{x_i}$ is a {\em constant} vector field on the linear space $Q_K$. Thus we constructed a map
\begin{equation}
\hat{\varphi}\mapsto \vec{X}(\hat{\varphi}):=\sum_i(\hat{\varphi}\kappa_i)\partial_{x_i}
\label{varphi-vecX}
\end{equation}
from the space $\hat{\cal F}$ to the linear space of all constant vector fields on $Q_K$. Clearly, the map is {\em linear}.

On the other hand there is a natural one-to-one correspondence between constant vector fields on the linear space $Q_K$ and points of $Q_K$---given constant vector field $\vec{X}$ on $Q_K$, there exists a unique $[q]\in Q_K$ such that for every differentiable function $\psi$ on $Q_K$  
\begin{equation}
(\vec{X}\psi)([q_0])=\frac{d}{dt}\Big|_{t=0}\psi([q_0]+t[q]).
\label{vecX-psi}
\end{equation}
It is also clear that this correspondence
\begin{equation}
\vec{X}\mapsto [q]
\label{vecX-q}
\end{equation}
is a linear isomorphism from the space of all constant vector fields on $Q_K$ onto $Q_K$.

Composing the maps \eqref{varphi-vecX} and \eqref{vecX-q} we obtain a {\em linear} map which associates with each $\hat{\varphi}\in\hat{\cal F}$ a point in $Q_K$. We will denote this point by $[\hat{\varphi}]$. Taking into account \eqref{vecX-psi} we can express \eqref{varphi-Psi} in the following form
\begin{equation}
(\hat{\varphi}\Psi)(q)=\frac{d}{dt}\Big|_{t=0}\psi([q]+t[\hat{\varphi}]).
\label{varphi-varphi}
\end{equation}   
We can say now that, given $\hat{\varphi}\in\hat{\cal F}$, $[\hat{\varphi}]$ is the unique point in  $Q_K$ such that \eqref{varphi-varphi} holds for every $q\in Q$ and for every cylindrical function $\Psi$ compatible with $K$.      

The formula \eqref{varphi-varphi} and the linear map 
\begin{equation}
\hat{\cal F}\ni\hat{\varphi}\to[\hat{\varphi}]\in Q_K
\label{map-vphi-vphi}
\end{equation}
just introduced will serve as tools which will be used below to define injections $\{\omega_{\lambda'\lambda}\}$ and to prove their properties.  

Let $K'$ be a finite set of independent d.o.f.. Then there exists a map
\begin{equation}
\hat{\cal F}\ni\hat{\varphi}\to[\hat{\varphi}]'\in Q_{K'}
\label{map-vphi-vphi'}
\end{equation}
defined analogously to \eqref{map-vphi-vphi}. Suppose that each elementary d.o.f. in $K$ is a linear combination of elementary d.o.f. in $K'$ which means that there exists a projection $\pr_{KK'}:Q_{K'}\to Q_K$. What is then a relation between $[\hat{\varphi}]'$ and $[\hat{\varphi}]$? To answer the question consider a cylindrical function $\Psi$ compatible with $K$. Then by virtue of \eqref{Psi-cyl} and \eqref{ppKKp}     
\[
\Psi=\pr^*_K\psi=\pr^*_{K'}(\pr^*_{KK'}\psi).
\]
Note that $\pr^*_{KK'}\psi$ is a function on $Q_{K'}$ and therefore $\Psi$ is compatible with $K'$ also. Thus Equation \eqref{varphi-varphi} can be expressed in the following way:
\begin{multline*}
(\hat{\varphi}\Psi)(q)=\frac{d}{dt}\Big|_{t=0}(\pr^*_{KK'}\psi)([q]'+t[\hat{\varphi}]')=\frac{d}{dt}\Big|_{t=0}\psi\Big(\pr_{KK'}([q]')+t\,\pr_{KK'}([\hat{\varphi}]')\Big)=\\=\frac{d}{dt}\Big|_{t=0}\psi\Big([q]+t\,\pr_{KK'}([\hat{\varphi}]')\Big),
\end{multline*}
where in the last step we used \eqref{prq'q}. Comparing this result with \eqref{varphi-varphi} we see that
\begin{equation}
\pr_{KK'}([\hat{\varphi}]')=[\hat{\varphi}].
\label{vphi'-vphi}
\end{equation}

Let us fix $\lambda\equiv(\hat{F},K)\in\Lambda$ and restrict the map \eqref{map-vphi-vphi} to $\hat{F}$. An important observation is that under Assumption \ref{non-deg} the restriction is a linear isomorphism. Indeed, let $(\hat{\varphi}_1,\ldots,\hat{\varphi}_N)$ be a basis of $\hat{F}$ and consider the action of the map \eqref{varphi-vecX} on the basis:
\[
\hat{\varphi}_j\mapsto \vec{X}(\hat{\varphi}_j)=\sum_{i=1}^N(\hat{\varphi}_j\kappa_i)\partial_{x_i}.
\]
By virtue of Assumption \ref{non-deg} $(\hat{\varphi}_j\kappa_i)=(G_{ji})$ is a non-degenerate matrix which means that the map \eqref{varphi-vecX} restricted to $\hat{F}$ is a linear isomorphism. Consequently, the map \eqref{map-vphi-vphi} restricted to  $\hat{F}$ is a linear isomorphism also. This means in particular that
\begin{equation}
[\hat{F}]=Q_K,
\label{F=QK}
\end{equation}
where $[\hat{F}]$ is the image of $\hat{F}$ under the map \eqref{map-vphi-vphi}.
  
Consider now an element $\lambda'\equiv(\hat{F}',K')\geq \lambda\equiv(\hat{F},K)$ and denote by $[\hat{F}]'\subset Q_{K'}$ the image of $\hat{F}$ under the map \eqref{map-vphi-vphi'}. Assume that for some $\hat{\varphi}\in\hat{F}$ the point $[\hat{\varphi}]'$ belongs to the kernel of ${\pr_{KK'}}$. Then by virtue of \eqref{vphi'-vphi}
\[
0=\pr_{KK'}([\hat{\varphi}]')=[\hat{\varphi}]\in Q_K.
\]         
But because  $\hat{\varphi}\in \hat{F}$ and, as shown above, $\hat{F}$ is linearly isomorphic to $Q_K$ we conclude that $\hat{\varphi}=0$ and therefore $[\hat{\varphi}]'=0$. Thus
\begin{equation}
\ker\pr_{KK'}\cap [\hat{F}]'=\{0\}.
\label{ker-F'}
\end{equation}
Consequently, $\pr_{KK'}$ restricted to $[\hat{F}]'$ is {\em injective}. It is also surjective---by virtue of \eqref{vphi'-vphi} and \eqref{F=QK}
\[
\pr_{KK'}([\hat{F}]')=[F]=Q_K.
\]
Thus $\pr_{KK'}$ restricted to $[\hat{F}]'\subset Q_{K'}$ is a linear isomorphism onto $Q_K$. 

Let us define then the desired map $\omega_{\lambda'\lambda}:Q_K\to Q_{K'}$ as the inverse map to $\pr_{KK'}$ restricted to $[\hat{F}]'$: 
\[
\omega_{\lambda'\lambda}:=\Big(\pr_{KK'}\Big|_{[\hat{F}]'}\Big)^{-1}
\] 
or, equivalently,
\begin{equation}
Q_K\ni[\hat{\varphi}]\mapsto\omega_{\lambda'\lambda}([\hat{\varphi}]):=[\hat{\varphi}]'\in Q_{K'}.
\label{omega}
\end{equation}
Obviously, the map $\omega_{\lambda'\lambda}$ is linear and injective and satisfies
\[
\pr_{KK'}\circ\omega_{\lambda'\lambda}={\rm id}.
\] 

Note that 
\[
\dim Q_{K'}=\dim\ker\pr_{KK'}+\dim Q_K=\dim\ker\pr_{KK'}+\dim [\hat{F}]'.
\]
Taking into account  \eqref{ker-F'} we conclude that
\begin{equation}
Q_{K'}=\ker\pr_{KK'}\oplus[\hat{F}]'=\ker\pr_{KK'}\oplus \omega_{\lambda'\lambda}(Q_K).
\label{QK'-dec}
\end{equation}

Consider now a triplet $\lambda''\equiv(\hat{F}'',K'')\geq\lambda'\equiv(\hat{F}',K')\geq \lambda\equiv(\hat{F},K)$ and the composition $\omega_{\lambda''\lambda'}\circ\omega_{\lambda'\lambda}$. The value of 
\[
(\omega_{\lambda''\lambda'}\circ\omega_{\lambda'\lambda})([\hat{\varphi}])
\]
for $\hat{\varphi}\in \hat{F}$ can be found as follows. First the map $\omega_{\lambda'\lambda}$ maps $[\hat{\varphi}]\in Q_{K}$ to $[\hat{\varphi}]'\in Q_{K'}$. Next, we find $\hat{\varphi}'\in\hat{F}'$ such that $[\hat{\varphi}']'=[\hat{\varphi}]'$. Then $\omega_{\lambda''\lambda'}$ maps $[\hat{\varphi}']'\in Q_{K'}$ to $[\hat{\varphi}']''\in Q_{K''}$. This mappings can be described as follows
\[
Q_{K}\ni [\hat{\varphi}]\mapsto [\hat{\varphi}]'=[\hat{\varphi}']'\mapsto[\hat{\varphi}']''=(\omega_{\lambda''\lambda'}\circ\omega_{\lambda'\lambda})([\hat{\varphi}])\in Q_{K''}.
\]       
Note now that due to Assumption \ref{FF'} $\hat{F}\subset\hat{F}'$ which means that $\hat{\varphi}'=\hat{\varphi}$ and consequently
\[
(\omega_{\lambda''\lambda'}\circ\omega_{\lambda'\lambda})([\hat{\varphi}])=[\hat{\varphi}]''=\omega_{\lambda''\lambda}([\hat{\varphi}]).
\]  
Thus
\begin{equation}
\omega_{\lambda''\lambda}=\omega_{\lambda''\lambda'}\circ\omega_{\lambda'\lambda}.
\label{omomom}
\end{equation}
  
Let us finally emphasize that, given a pair $\lambda'\equiv(\hat{F}',K')\geq \lambda\equiv(\hat{F},K)$, we denoted the corresponding injection by $\omega_{\lambda'\lambda}$ despite the fact that the space $\hat{F}'$ was not used to construct the map---necessary and sufficient ingredients for the construction of $\omega_{\lambda'\lambda}$ are sets $(\hat{F},K)$ satisfying Assumptions \ref{const} and \ref{non-deg} and a set $K'$ of independent d.o.f. such that each elementary d.o.f. in $K$ is a linear combination of elementary d.o.f. in $K'$. Note also that Assumption \ref{FF'} was necessary merely for proving the property \eqref{omomom}.

%***************************************************
\subsubsection{Step 3}
%***************************************************

In this step we will define a Hilbert space $\h_\lambda$ associated with $\lambda\equiv(\hat{F},K)\in\Lambda$. 

Given reduced configuration space $Q_K$ generated by $K\in\lambda$, due to Assumption \ref{RN} there exists on it a distinguished linear coordinate frame $(x_i)$ defined by \eqref{lin-coor}. This frame defines a measure on $Q_K$ 
\begin{equation}
d\mu_{\lambda}:=dx_1\ldots dx_N,
\label{dmu-la}
\end{equation}
where $N=\dim Q_K$. In other words, $d\mu_\lambda$ is the Lebesgue measure on $\R^N$ pulled back by the map $\tilde{K}$ (see \eqref{k-inj}) to $Q_K$. {Note that $Q_K$ as a finite dimensional linear space is an Abelian Lie group and $d\mu_{\lambda}$ being invariant with respect to translations on $Q_K$ is a Haar measure on $Q_K$.}  

Let us define a Hilbert space $\h_\lambda$ as a space of complex functions on $Q_K$ square integrable with respect to the measure $d\mu_\lambda$,
\begin{equation}
\h_\lambda:=L^2(Q_K,d\mu_\lambda).
\label{h-la}
\end{equation}

Of course, neither the measure $d\mu_\lambda$ nor the Hilbert space $\h_\lambda$ depends on $\hat{F}\in\lambda$, so it would be perhaps more natural to denote them by, respectively, $d\mu_K$ and $\h_K$, but the chosen symbols $d\mu_\lambda$ and $\h_\lambda$ will be more convenient for further applications.

As mentioned already, the space $\h_\la$ represents pure quantum states of the system $\la$.

%***************************************************
\subsubsection{Step 4}
%***************************************************

Denote by $\D_\lambda$ the space of all density operators (i.e. positive operators of trace equal $1$) on the Hilbert space $\h_\lambda$. The goal of this step is to define  the projection $\pi_{\lambda'\lambda}:\D_{\lambda'}\to \D_\lambda$ for every $\lambda'\geq\lambda$.

Assume then that $\la'\equiv(\hat{F}',K')\geq\la\equiv(\hat{F},K)$. Let us recall that the task of the projection $\pi_{\la\la'}$ is to reduce quantum d.o.f. of the system $\la'$ to those of its subsystem $\la$. Thus we will decompose the Hilbert space $\h_{\la'}$ into a tensor product of a Hilbert space $\tilde{\h}_{\la'\la}$ describing quantum d.o.f. which should be reduced and a Hilbert space $\h_{\la'\la}$ corresponding to quantum d.o.f. of the subsystem $\la$:
\begin{equation}
\h_{\la'}=\tilde{\h}_{\la'\la}\ot \h_{\la'\la}
\label{hhh-0}
\end{equation}
and then will define the projection $\pi_{\la\la'}$ by means of the partial trace with respect to the Hilbert space $\tilde{\h}_{\la'\la}$.  

A natural way to obtain such a decomposition of $\h_{\la'}$ is to derive it from a decomposition of the space $Q_{K'}$ which underlies the Hilbert space (see \eqref{h-la}). We choose the following decomposition 
\begin{equation}
Q_{K'}=\ker\pr_{KK'}\oplus Q_{\la'\la}
\label{QK'-dec-1}
\end{equation}
and will define $\tilde{\h}_{\la'\la}$ as a Hilbert space of functions on $\ker\pr_{KK'}$ and $\h_{\la'\la}$ as a Hilbert space of functions on the linear space $Q_{\la'\la}$. To justify this choice of the decomposition let us recall that $\pr_{KK'}$ reduces the configurational d.o.f. of $Q_{K'}$ to those of $Q_K$ hence $\ker\pr_{KK'}$ is constituted from the configurational d.o.f. which undergo to the reduction. On the other hand, $\pr_{KK'}$ restricted to $Q_{\la'\la}$ is injective which means that $Q_{\la'\la}$ contains those configurational d.o.f. which survive the reduction.

Note however that the decomposition \eqref{QK'-dec-1} is not unique i.e. the space $Q_{\la'\la}$ in \eqref{QK'-dec-1} is not unique. As it was shown in \cite{oko-ncomp} projections $\{\pi_{\la\la'}\}$ constructed according to the prescription just presented  do depend on the choice of $Q_{\la'\la}$. This may seem to be a problem but in fact it is not---the passage from the system $\la'$ to its subsystem $\la$ involves a reduction of {\em both} momentum and configurational d.o.f of $\la'$ and the decomposition \eqref{QK'-dec-1} takes into account only the reduction of the configurational ones. So the freedom to choose $Q_{\la'\la}$ corresponds to the freedom to choose which linear subspace of $\hat{F}'$ is the space $\hat{F}$ of momentum d.o.f. of the subsystem $\la$. This means that to construct the projections $\{\pi_{\la\la'}\}$ we have to find a space $Q_{\la'\la}$ corresponding to $\hat{F}$.

It seems that a perfect candidate to serve as the space $Q_{\la'\la}$ corresponding to $\hat{F}$ is the image of the injection $\omega_{\la'\la}$---recall that the image corresponds naturally to $\hat{F}$ and that it satisfies \eqref{QK'-dec} which is a particular case of the decomposition \eqref{QK'-dec-1}. Thus we set 
\[
Q_{\la'\la}=\omega_{\la'\la}(Q_K)
\]
and rewrite the decomposition \eqref{QK'-dec-1} in the following form
\begin{equation}
\begin{array}{ccccc}
Q_{K'}&=&\ker \pr_{KK'}&\oplus&\omega_{\la'\la}(Q_K)\medskip\\
&&&&\Big\uparrow\vcenter{\rlap{$\scriptstyle{\omega_{\la'\la}}$}}\medskip\\
&&&&Q_K
\end{array}\ \ \,.
\label{QK'-dec->}
\end{equation}

We define now the Hilbert spaces
\begin{align}
&\tl{\h}_{\la'\la}:=L^2(\ker\pr_{KK'},d\tl{\mu}_{\la'\la}),&& \h_{\la'\la}:=L^2(\omega_{\la'\la}(Q_K),d\mu_{\la'\la}),
\label{hh}
\end{align}
where the measure $d\mu_{\la'\la}$ on $\omega_{\la'\la}(Q_K)$ is given by the push-forward
\begin{equation}
d\mu_{\la'\la}:=\omega_{\la'\la*}d\mu_{\la}.
\label{mu-ll}
\end{equation}
To define the measure $d\tl{\mu}_{\la'\la}$ on $\ker\pr_{KK'}$ we need to take into account that in some cases $\ker\pr_{KK'}$ is zero-dimensional i.e. $\ker\pr_{KK'}=\{0\}$, where $0$ is the zero of the vector space $Q_{K'}$ and to define first what is a measure on such a degenerate space. If $\ker\pr_{KK'}=\{0\}$ then one can define a measure $d\mu$ on it requiring that for every complex function $f$ on $\ker\pr_{KK'}$ 
\[
\int_{\ker\pr_{KK'}}f\,d\mu:=f(0)\,\xi\in\C,
\]
where $\xi$ is a {\em real positive} number {\em independent} of $f$ (one can say that the measure $d\mu$ is just the number $\xi$). Then the Hilbert space $L^2(\ker\pr_{KK'},d{\mu})$ is naturally isomorphic to a Hilbert space $\C$ of complex numbers equipped with the scalar product
\[
\C^2\ni(z,z')\mapsto\scal{z}{z'}:=\bar{z}z'\xi\in\C.
\]         

Now it is easy to see that there exists a unique measure $d\tl{\mu}_{\la'\la}$ on  $\ker\pr_{KK'}$ such that $d\mu_{\la'}=d\tl{\mu}_{\la'\la}\times d\mu_{\la'\la}$. Thus we have obtained the following diagram  
\begin{equation}
\begin{array}{ccccc}
d\mu_{\la'}&=&d\tl{\mu}_{\la'\la}&\times&d\mu_{\la'\la}\medskip\\
&&&&\Big\uparrow\vcenter{\rlap{$\scriptstyle{\omega_{\la'\la*}}$}}\medskip\\
&&&& d\mu_{\la}
\end{array}\ \ \,,
\label{mmm}
\end{equation}
which corresponds to \eqref{QK'-dec->}. 

It follows from the formulae \eqref{hh} and \eqref{mu-ll} defining, respectively, the Hilbert space $\h_{\la'\la}$ and the measure $d\mu_{\la'\la}$ that the following map 
\begin{equation}
\h_{\la'\la}\ni\Psi\mapsto U_{\la'\la}\Psi:=\omega_{\la'\la}^{*}\Psi\in\h_\la,
\label{Ull}
\end{equation}
is unitary and therefore it can be used to identify the Hilbert spaces $\h_{\la'\la}$ and $\h_\la$.

Now we can rewrite the decomposition \eqref{hhh-0} in the following form
\begin{equation}
\begin{array}{ccccc}
\h_{\la'}&=&\tl{\h}_{\la'\la}&\ot&\h_{\la'\la}\medskip\\
&&&&\Big\downarrow\vcenter{\rlap{$\scriptstyle{U_{\la'\la}=\,\omega^{*}_{\la'\la}}$}}\medskip\\
&&&& \h_\la
\end{array}\ \ \ \ \ \ \ \ \,.
\label{hhh}
\end{equation}
which corresponds to the diagrams \eqref{QK'-dec->} and \eqref{mmm}.

Now we are able to define the projection $\pi_{\la'\la}$. Given $\rho\in\D_{\la'}$, we act on it by an operator $\tl{\tr}_{\la'\la}$ of the partial trace with respect to the Hilbert space $\tl{\h}_{\la'\la}$ obtaining thereby a density operator $\tr_{\la'\la}\rho$ on the Hilbert space $\h_{\la'\la}$. Then we map the resulting density operator to a density operator on $\h_\la$ by an isomorphism $u_{\la'\la}$ from the algebra ${\cal B}_{\la'\la}$ of bounded operators on $\h_{\la'\la}$ onto the algebra ${\cal B}_\la$ of bounded operators on $\h_\la$ given by
\begin{equation}
{\cal B}_{\la'\la}\ni \alpha \mapsto u_{\la'\la}\alpha:=U_{\la'\la}\circ \alpha \circ U^{-1}_{\la'\la}\in{\cal B}_\la.
\label{ull}
\end{equation}
Thus the composition of $\tl{\tr}_{\la'\la}$ and $u_{\la'\la}$ projects density operators in $D_{\la'}$ to ones in $D_\la$:       
\begin{df}
Given $\la'\geq\la$, the projection $\pi_{\la\la'}:\D_{\la'}\rightarrow\D_\la$ is defined by the following formula:
\[
\pi_{\la\la'}:=
u_{\la'\la}\circ\tl{\tr}_{\la'\la}. 
\]
\label{forget}
\end{df}

\begin{pro}
For every triplet $\la''\geq\la'\geq\la$ of elements of $\Lambda$ 
\begin{equation}
\pi_{\la\la''}=\pi_{\la\la'}\circ\pi_{\la'\la''}.
\label{pipipi}
\end{equation}
\end{pro}

\begin{proof}
It will be convenient to describe the partial trace $\tl{\tr}_{\la'\la}$ using the integral kernel $\check{\rho}$ of $\rho$ being a function on $Q^2_{K'}$ square integrable with respect to the measure $d\mu_{\la'}\times d\mu_{\la'}$. The kernel of the resulting operator $\tl{\tr}_{\la'\la}\rho$ is a function on $(\omega_{\la'\la}(Q_K))^2$ and reads
\[
(\tl{\tr}_{\la'\la}\check{\rho})(b',b):=\int_{\ker\pr_{KK'}} \check{\rho}(a+b',a+b)\,\,d\tl{\mu}_{\la'\la}(a),
\]  
where $a\in\ker\pr_{KK'}$ and $b',b\in\omega_{\la'\la}(Q_K)$. To obtain the kernel of the density operator $\pi_{\la\la'}\rho$ we have to pull back the kernel just obtained to $Q^2_K$ by means of the injection  $\omega_{\la'\la}$---recall that this injection defines via \eqref{Ull} the unitary map $U_{\la'\la}$ which on the other hand defines via \eqref{ull} the isomorphism $u_{\la'\la}$ mapping $\tl{\tr}_{\la'\la}\rho$ to $\pi_{\la\la'}\rho$. Thus the kernel of $\pi_{\la\la'}\rho$ can be expressed as
\begin{multline}
(\pi_{\la\la'}\check{\rho})(b',b)=(\tl{\tr}_{\la'\la}\check{\rho})\Big(\omega_{\la'\la}(b'),\omega_{\la'\la}(b)\Big)=\\=\int_{\ker\pr_{KK'}} \check{\rho}\Big(a+\omega_{\la'\la}(b'),a+\omega_{\la\la'}(b)\Big)\,\,d\tl{\mu}_{\la'\la}(a),
\label{pi-ker-r}
\end{multline}
where $a\in\ker\pr_{KK'}$ and $b,b'\in Q_K$. 

Let $\la''\geq\la'\geq\la$. Then using twice Equation \eqref{pi-ker-r} we obtain the following formula describing the kernel of $\pi_{\la\la'}\pi_{\la'\la''}\rho$:  
\begin{multline*}
(\pi_{\la\la'}\pi_{\la'\la''}\check{\rho})(c',c)=\\=\int\limits_{\substack{\ker\\\pr_{KK'}}}\Big[\int\limits_{\substack{\ker\\\pr_{K'K''}}}\check{\rho}\Big(a+\omega_{\la''\la'}(b+\omega_{\la'\la}(c')),a+\omega_{\la''\la'}(b+\omega_{\la'\la}(c))\Big)\,\,d\tl{\mu}_{\la''\la'}(a)\,\Big]\,\,d\tl{\mu}_{\la'\la}(b),
\end{multline*}
where $a\in\ker\pr_{K'K''}$, $b\in \ker\pr_{KK'}$ and $c',c\in Q_K$. By virtue of linearity of $\omega_{\la''\la'}$ and the property \eqref{omomom} we simplify the formula above obtaining:
\begin{multline*}
(\pi_{\la\la'}\pi_{\la'\la''}\check{\rho})(c',c)=\\=\int\limits_{\substack{\ker\\\pr_{KK'}}}\Big[\int\limits_{\substack{\ker\\\pr_{K'K''}}}\check{\rho}\Big(a+\omega_{\la''\la'}(b)+\omega_{\la''\la}(c'),a+\omega_{\la''\la'}(b)+\omega_{\la''\la}(c)\Big)\,\,d\tl{\mu}_{\la''\la'}(a)\,\Big]\,\,d\tl{\mu}_{\la'\la}(b).
\end{multline*} 
Let us change in this integral the variable $b$ to
\[
\bar{b}:=\omega_{\la''\la'}(b)\in \omega_{\la''\la'}({\ker\pr_{KK'}})\subset Q_{K''}.
\]
Then 
\begin{multline}
(\pi_{\la\la'}\pi_{\la'\la''}\check{\rho})(c',c)=\\=\int\limits_{\omega_{\la''\la'}(\ker\pr_{KK'})}\Big[\int\limits_{\ker\pr_{K'K''}}\check{\rho}\Big(a+\bar{b}+\omega_{\la''\la}(c'),a+\bar{b}+\omega_{\la''\la}(c)\Big)\,\,d\tl{\mu}_{\la''\la'}(a)\,\Big]\,\,d\tl{\mu}_{\la''\la'\la}(\bar{b}),
\label{pipirho}
\end{multline}
where     
\begin{equation}
d\tl{\mu}_{\la''\la'\la}:=\omega_{\la''\la'*}d\tl{\mu}_{\la'\la}
\label{mu-lll}
\end{equation}
is a measure on $\omega_{\la''\la'}({\ker\pr_{KK'}})\subset Q_{K''}$.

Note now that domains of both integrals in the formula \eqref{pipirho} are linear subspaces of $Q_{K''}$. Applying the decomposition \eqref{QK'-dec} first to $Q_{K''}$ and then to $Q_K'$ we obtain
\begin{multline*}
Q_{K''}=\ker\pr_{K'K''}\oplus \omega_{\la''\la'}\Big(\ker\pr_{KK'}\oplus\omega_{\la'\la}(Q_K)\Big)=\\=\ker\pr_{K'K''}\oplus \omega_{\la''\la'}(\ker\pr_{KK'})\oplus\omega_{\la''\la}(Q_K),
\end{multline*}
where in the last step we used \eqref{omomom}. Suppose now that $\bar{a}\in \ker\pr_{KK''}\subset Q_{K''}$. Due to \eqref{pr-KKK} 
\begin{equation}
\bar{a}=a+\bar{b},
\label{a=a+b}
\end{equation}
where  $a\in\ker\pr_{K'K''}$ and $\pr_{K'K''}(\bar{b})\in\ker\pr_{KK'}$ or, equivalently, $\bar{b}\in\omega_{\la''\la}(\ker\pr_{KK'})$. Therefore
\begin{equation}
\ker\pr_{K'K''}\oplus \omega_{\la''\la'}(\ker\pr_{KK'})=\ker\pr_{KK''}.
\label{kerrr}
\end{equation}
On the other hand by virtue of \eqref{mmm} applied to $d\mu_{\la''}$ and $d\mu_{\la'}$ 
\begin{multline*}
d\mu_{\la''}=d\tl{\mu}_{\la''\la'}\times \omega_{\la''\la'*}(d\mu_{\la'})=d\tl{\mu}_{\la''\la'}\times \omega_{\la''\la'*}\Big(d\tl{\mu}_{\la'\la}\times \omega_{\la'\la*}(d\mu_{\la})\Big)=\\=d\tl{\mu}_{\la''\la'}\times \omega_{\la''\la'*}(d\tl{\mu}_{\la'\la})\times \omega_{\la''\la*}(d\mu_{\la})=d\tl{\mu}_{\la''\la'}\times d\tl{\mu}_{\la''\la'\la}\times \omega_{\la''\la*}(d\mu_{\la}),
\end{multline*}
where we used \eqref{omomom} and \eqref{mu-lll}. The decomposition \eqref{mmm} allows us to conclude that
\begin{equation}
d\tl{\mu}_{\la''\la'}\times d\tl{\mu}_{\la''\la'\la}=d\tl{\mu}_{\la''\la}.
\label{mumumu}
\end{equation}

Using the results \eqref{a=a+b}, \eqref{kerrr} and \eqref{mumumu} we rewrite the integral \eqref{pipirho} obtaining
\begin{equation*}
(\pi_{\la\la'}\pi_{\la'\la''}\check{\rho})(c',c)=\int\limits_{\ker\pr_{KK''}}\check{\rho}\Big(\bar{a}+\omega_{\la''\la}(c'),\bar{a}+\omega_{\la''\la}(c)\Big)\,\,d\tl{\mu}_{\la''\la}(\bar{a}),
\end{equation*}
where $\bar{a}\in\ker\pr_{KK''}$ and $c',c\in Q_K$. Comparing this with \eqref{pi-ker-r} we arrive at
\[
(\pi_{\la\la'}\pi_{\la'\la''}\check{\rho})(c',c)=(\pi_{\la\la''}\check{\rho})(c',c)
\]  
which is obviously equivalent to \eqref{pipipi}.
\end{proof}

%***************************************************
\subsubsection{Step 5}
%***************************************************

Thus we have constructed the family $\{\D_\la,\pi_{\la\la'}\}_{\la\in\Lambda}$ and proved that the projections $\{\pi_{\la\la'}\}$ satisfy the consistency condition \eqref{pipipi}. Thus $\{\D_\la,\pi_{\la\la'}\}_{\la\in\Lambda}$ is a {\em projective family}. This fact allow us to define the space $\D$ of quantum states of the theory under consideration as the projective limit of the family 
\begin{equation}
\D:=\underleftarrow{\lim} \,\D_\lambda.
\label{D-df}
\end{equation}

%***************************************************
\subsection{Remarks on Assumptions}
%***************************************************

After finishing the construction we can clearly see the role played by each of Assumptions described in Section \ref{ad-as}. Thus Assumptions \ref{k-Lambda} and \ref{f-Lambda} ensure that the set $\Lambda$ is rich enough to include information encoded in all elementary d.o.f.. Moreover, Assumption \ref{k-Lambda} assures that each configurational elementary d.o.f. is a cylindrical function. Assumption \ref{RN} enables to introduce a structure of a real linear space on reduced configuration spaces $\{Q_K\}$. Assumptions \ref{comp-f} and \ref{const} guarantee that each operator $\hat{\varphi}$ defines a constant vector field on $Q_K$ which is an important element of the construction of injections $\{\omega_{\la'\la}\}$. Note that Assumption \ref{comp-f} is naturally satisfied if operators $\{\hat{\varphi}\}$ are defined via Poisson bracket \eqref{hat-varphi} without any  regularization. On the other hand Assumption \ref{const} makes it easy to formulate Assumption \ref{non-deg} and the latter assumption plays an important role in the construction of the injections $\{\omega_{\la'\la}\}$. Assumption \ref{Q'=Q} seems to be a natural consistency requirement imposed on the relation $\geq$ motivated by the fact that if $(\hat{F},K'),(\hat{F},K)\in\Lambda$ and $Q_{K'}=Q_{K}$ then $(\hat{F},K')$ and $(\hat{F},K)$ describe the same reduced system. Assumptions \ref{lin-comb} and \ref{FF'} allow to interpret $(\hat{F},K)$ as a subsystem of $(\hat{F}',K')$. Moreover, Assumption \ref{lin-comb} guarantees linearity of projections $\{\pr_{KK'}\}$.

%***************************************************
\subsection{Construction of a directed set $(\Lambda,\geq)$---auxiliary facts \label{aux}}
%***************************************************

Here we will present and prove some auxiliary facts which may be very useful while constructing a directed set $(\Lambda,\geq)$ for a theory (see Section \ref{DPG} and \cite{q-stat}).

\begin{lm}
Let $K,K'$ be sets of independent d.o.f. of $N$ and $N'$ elements respectively. Suppose that the image of the map $\tilde{K}'$  given by \eqref{k-inj} is $\R^{N'}$  and each d.o.f. in $K$ is a linear combination of d.o.f. in $K'$. Then the image of $\tilde{K}$ is $\R^N$.
\label{KK'-lm}
\end{lm} 

\begin{proof}
Let $K=\{\kappa_1,\ldots,\kappa_N\}$ and $K'=\{\kappa'_1,\ldots,\kappa'_{N'}\}$. There exists a constant matrix $B=(B_i{}^j)$, ($i=1,\ldots,N$; $j=1,\ldots,N'$), such that for every $q\in Q$ 
\[
\kappa_{i}(q)=B_i{}^j\kappa'_{j}(q)
\]     
or, equivalently,
\[
\tilde{K}([q])=B\tilde{K}'([q]'),
\]
where $[q]\in Q_K$ and $[q]'\in Q_{K'}$. Since the image of $\tilde{K}'$ is $\R^{N'}$ the image of $\tilde{K}$ is $\R^M$ where $M\leq N$ is the rank of the matrix $B$. But we assumed that $K$ is a set of independent d.o.f. which means that the image of $\tilde{K}$ is an $N$-dimensional submanifold of $\R^N$. Thus the rank $M$ of $B$ is maximal and  $M=N$.           
\end{proof}

\begin{pro}
Let $K,K'$ be sets of independent d.o.f. of $N$ and $N'$ elements respectively such that $Q_K=Q_{K'}$. Suppose that there exists a set $\bar{K}$ of independent d.o.f. of $\bar{N}$ elements such that the image of $\tilde{\bar{K}}$ is $\R^{\bar{N}}$ and  each d.o.f. in $K\cup K'$ is a linear combination of d.o.f. in $\bar{K}$. Then $\tilde{K}$, $\tilde{K}'$ are bijections, $N=N'$, $\tilde{K}\circ\tilde{K}^{\prime -1}$ is a linear automorphism from $\R^N$ onto itself and each d.o.f. in $K$ is a linear combination of d.o.f. in $K'$.  
\label{KK'barK}
\end{pro}

\begin{proof}
We know from Lemma \ref{KK'-lm} that the images of $\tilde{K}$ and $\tilde{K}'$ are $\R^N$ and $\R^{N'}$ respectively, i.e. the maps are bijections. There exist constant matrices $B$ and $B'$ such that (see the proof of Lemma \ref{KK'-lm})
\begin{align}
&\tilde{K}([q])=B\tilde{\bar{K}}(\overline{[q]}),,&&\tilde{K}'([q]')=B'\tilde{\bar{K}}(\overline{[q]}),
\label{KBK2}
\end{align}
where $[q]\in Q_K$, $[q]'\in Q_{K'}$ and $\overline{[q]}\in Q_{\bar{K}}$. We assumed that $Q_{K}=Q_{K'}$, that is, $[q]=[q']$ for every $q\in Q$. Since $\tilde{K}'$ is bijection the second equation in \eqref{KBK2} can be expressed in the following form:
\[
[q]=[q]'=\tilde{K}^{\prime -1}\Big(B'\tilde{\bar{K}}(\overline{[q]})\Big).
\]       
Setting this to the first equation in \eqref{KBK2} we obtain
\[
\Big(\tilde{K}\circ\tilde{K}^{\prime -1}\circ (B'\tilde{\bar{K}})\Big)(\overline{[q]})=B\tilde{\bar{K}}(\overline{[q]}).
\]
Taking into account  that $\tilde{\bar{K}}$ is a bijection we rewrite the equation above as follows
\[
(\tilde{K}\circ\tilde{K}^{\prime -1})\circ B'=B
\]   
--- here we treat $B,B'$ as linear surjections from $\R^{\bar{N}}$  onto, respectively, $\R^N$ and $\R^{N'}$; the maps are surjections because, as we know from the proof of Lemma \ref{KK'-lm}, the rank of both matrices is maximal.

The above equation and the properties of $B$ and $B'$  imply that $\tilde{K}\circ\tilde{K}^{\prime -1}$ is a linear map from $\R^{N'}$ onto $\R^N$. Since it is a composition of two bijections it is also a bijection. Thus $N=N'$ and $\tilde{K}\circ\tilde{K}^{\prime -1}$ is a linear automorphism on $\R^N$.

Moreover, the following identity 
\[
\tilde{K}([q])=(\tilde{K}\circ\tilde{K}^{\prime -1})\tilde{K}^{\prime -1}([q]')
\]
which holds for every $q\in Q$ can be rewritten as (see \eqref{k-inj})
\[
\begin{pmatrix}
\kappa_1(q)\\
\vdots\\
\kappa_N(q)
\end{pmatrix}=(\tilde{K}\circ\tilde{K}^{\prime -1})
\begin{pmatrix}
\kappa'_1(q)\\
\vdots\\
\kappa'_N(q)
\end{pmatrix},
\]    
which means that each d.o.f. $\kappa_i$ in $K$ is a linear combination of d.o.f. $\{\kappa'_1,\ldots,\kappa'_N\}=K'$.     
\end{proof}

\begin{pro}
Suppose that there exists a subset $\mathbf{K}'$ of $\mathbf{K}$ such that  for every finite set $K_0$ of configurational elementary d.o.f. there exists $K'_0\in\mathbf{K}'$ satisfying the following conditions: 
\begin{enumerate}
\item the map $\tilde{K}'_0$ is a bijection; 
\item each d.o.f. in $K_0$ is a linear combination of d.o.f. in $K'_0$.  
\end{enumerate} 
Then
\begin{enumerate}
\item for every set $K\in\mathbf{K}$ the map $\tilde{K}$ is a bijection. Consequently, $Q_K\cong \R^N$ with $N$ being the number of elements of $K$ and the map $\tilde{K}$ defines a linear structure on $Q_K$ being the pull-back of the linear structure on $\R^N$; if $Q_{K}=Q_{K'}$ for some other set $K'\in\mathbf{K}$ then the linear structures defined on the space by $\tilde{K}$ and $\tilde{K}'$ coincide.
\item if a cylindrical function $\Psi$ compatible with a set $K\in\mathbf{K}$ can be expressed as
\[
\Psi=\pr_{K'}\psi',
\]  
where $K'\in\mathbf{K}$ and $\psi'$ is a complex function on $Q_{K'}$ then $\psi'$ is smooth and consequently $\Psi$ is compatible with $K'$;     
\item for every element $\Psi\in\Cyl$ there exists a set $K\in\mathbf{K}'$ such that $\Psi$ is compatible with $K$.    
\end{enumerate}
\label{big-pro}
\end{pro}
\noindent Note that the first assertion of the proposition means that on every reduced configuration space there exists a natural differential structure which guarantees that the space $\Cyl$ is well defined. 

\begin{proof}\mbox{}
  \noindent\paragraph{Assertion 1} Assume that $K_0$ is a set of independent d.o.f. also. Then we can apply Lemma \ref{KK'-lm} to $K_0$ and $K_0'$ and conclude that the image of $\tilde{K}_0$ is $\R^{N_0}$, where $N_0$ is the number of elements of $K_0$. Thus $\tilde{K}_0$ is a bijection which can be used to pull back the linear structure on $\R^{N_0}$ onto $Q_{K_0}$. The conclusion is that on every reduced configuration space there exists a {\em linear structure}.

Suppose now that for distinct sets of independent d.o.f. ${K}$ and ${K}'$ the spaces $Q_{K}$ and $Q_{{K}'}$ coincide, $Q_K=Q_{K'}$. The assumptions of the proposition under consideration guarantee that there exists a set $\bar{K}$ of independent d.o.f. such that Proposition \ref{KK'barK} can be applied to the sets $K,K'$ and $\bar{K}$. By virtue of that proposition $\tilde{K}\circ\tilde{K}^{\prime -1}$ is a linear automorphism from $\R^N$ onto itself (where $N$ is the number of elements of $K$ and of elements of $K'$) the linear structure on $Q_K=Q_{K'}$ defined by $\tilde{K}$ coincides with one defined by $\tilde{K}'$. 

Thus on every reduced configuration space there exists a {\em natural} linear structure.   
\paragraph{Assertion 2} Let $\bar{K}$ be a set of independent d.o.f. such that each d.o.f. in $K\cup K'$ is a linear combination of d.o.f. in $\bar{K}$. We know already that the maps $\tilde{K},\tilde{K}'$ and $\tilde{\bar{K}}$ are bijections. Thus by virtue of Lemma \ref{lm-pr} there exist linear projections 
\begin{align*}
&\pr_{K\bar{K}}:Q_{\bar{K}}\to Q_{K}, && \pr_{K'\bar{K}}:Q_{\bar{K}}\to Q_{K'} 
\end{align*}
such that
\[
\Psi=\pr^*_{\bar{K}}(\pr^*_{K\bar{K}}\psi)=\pr^*_{\bar{K}}(\pr^*_{K'\bar{K}}\psi').
\]

Assume that the function $\psi'$ is not smooth. Then $\pr^*_{K'\bar{K}}\psi'$ (being a complex function on $Q_{\bar{K}}$) is not smooth either. But $\pr^*_{K\bar{K}}\psi$ is smooth and
\[
\pr^*_{K'\bar{K}}\psi'=\pr^*_{K\bar{K}}\psi.
\]   
Thus we see that the assumption that $\psi'$ is not smooth leads to a contradiction.

Since $\psi'$ is smooth $\Psi$ is a cylindrical function compatible with $K'$.      

\paragraph{Assertion 3} 

Consider an arbitrary function $\Psi\in\Cyl$. It is a finite sum of cylindrical functions:
\begin{equation}
\Psi=\sum_{a=1}^n \Psi_a=\sum_{a=1}^n\pr^*_{K_a}\psi_a,
\label{Psi=sum}
\end{equation}
where $\Psi_a$ is a cylindrical function compatible with a set $K_a$ of independent d.o.f, and $\psi_a$ is a complex function on $Q_{K_a}$. Let $\bar{K}\in\mathbf{K}'$ be such that each d.o.f. in $\bigcup_{a=1}^n K_a$ is a linear combination of d.o.f. in $\bar{K}$. Using Lemma \ref{lm-pr} we can write
\[
\Psi=\pr^*_{\bar{K}}\Big(\sum_{a=1}^n\pr^*_{\bar{K}K_a}\psi_a\Big),
\] 
which means that $\Psi$ is compatible with $\bar{K}$.

\end{proof}

Let $\bld{\Psi}$ be a subset of $\Cyl$. Then operators in $\hat{\F}$ restricted to $\bld{\Psi}$ can be regarded as maps from $\bld{\Psi}$ into $\Cyl$. Recall that both $\Cyl$ and $\hat{\F}$ are linear spaces. Therefore the restricted operators can be regarded as maps valued in a linear space and the space of all the restricted operators is a linear space. Consequently, the notion of linear independence of the restricted operators  is well defined---below this notion will be often called a linear independence of the operators on $\bld{\Psi}$.     

\begin{lm}
Let $\Cyl_K$ be a set of all cylindrical functions compatible with a set $K$ of independent d.o.f..  Assume that operators $\{\hat{\varphi}_1,\ldots,\hat{\varphi}_M\}\subset \hat{\F}$ act on elements of $\Cyl_K$ according to the formula in Assumption \ref{comp-f}. If $\{\hat{\varphi}_1,\ldots,\hat{\varphi}_M\}\subset \hat{\F}$ are linearly independent on a subset $\bld{\Psi}$ of $\Cyl_K$ then they are linearly independent on $K$.       
\label{cyl-K}
\end{lm}

\begin{proof}
Suppose that the operators $\{\hat{\varphi}_1,\ldots,\hat{\varphi}_M\}$ are linearly dependent on the set $K=\{\kappa_1,\ldots,\kappa_N\}$. This means that there exists real numbers $a_1,\ldots,a_M$ such that $a^2_1+\ldots+a^2_M\neq 0$ and for every $\kappa_i\in{\cal K}$ 
\begin{equation}
\sum_ja_j\hat{\varphi}_j\kappa_i=0.
\label{a-varphi-0}
\end{equation} 

Let $\Psi=\pr^*_K\psi$ be an arbitrary element of $\bld{\Psi}$. Since the operators act on cylindrical functions in $\Cyl_K$ according to the formula in Assumption \ref{comp-f}
\begin{equation}
\sum_ja_j\hat{\varphi}_j\Psi=\sum_{i}\Big(\pr^*_{K}\partial_{x_{i}}\psi\Big)\Big( \sum_ja_j\hat{\varphi}_j\kappa_{i}\Big)=0,
\label{aphiPsi-0}
\end{equation}
where in the last step we used \eqref{a-varphi-0}. 

Because $\Psi$ is arbitrary we conclude that the operators under consideration are linearly dependent on $\bld{\Psi}$. Thus if the operators are linearly independent on $\bld{\Psi}$ the they are linearly independent on $K$.     
\end{proof}

\begin{pro}
Let $\Lambda$ be a subset of $\hat{\mathbf{F}}\times\mathbf{K}$ which satisfies Assumptions \ref{k-Lambda} and \ref{comp-f}. Then for every finite set $\{\hat{\varphi}_1,\ldots,\hat{\varphi}_M\}\subset\hat{\cal F}$ of linearly independent operators there exists a set $(\hat{F},K)\in\Lambda$ such that the operators restricted to $K$ remain linearly independent.
\label{Lambda-pr}   
\end{pro} 

\begin{proof}
The proof consists of three steps.

\paragraph{Step 1} Suppose that $\{\hat{\varphi}_1,\ldots,\hat{\varphi}_M\}\subset\hat{\cal F}$ are linearly dependent when restricted to ${\cal K}$. This means that there exists real numbers $a_1,\ldots,a_M$ such that $a^2_1+\ldots+a^2_M\neq 0$ and for every $\kappa\in{\cal K}$ 
\begin{equation}
\sum_ja_j\hat{\varphi}_j\kappa=0.
\label{a-varphi}
\end{equation}

Let $\Psi$ be an arbitrary element of $\Cyl$. Then it can be expressed by means of the formula \eqref{Psi=sum}. It follows from Assumption \ref{k-Lambda} that there exists $(\hat{F},K)\in\Lambda$ such that each d.o.f. in $\bigcup_{a=1}^n K_a$ ($\{K_a\}$ are reduced configuration spaces appearing in \eqref{Psi=sum}) is a cylindrical function compatible with $K$. Consequently, $\Psi$ is compatible with $K$ and by virtue of Assumption \ref{comp-f} and \eqref{a-varphi} the formula \eqref{aphiPsi-0} holds. Since $\Psi$ is an arbitrary element of $\Cyl$ this result means that the operators $\{\hat{\varphi}_1,\ldots,\hat{\varphi}_M\}$ are linearly dependent on $\Cyl$. 

We conclude that if $\{\hat{\varphi}_1,\ldots,\hat{\varphi}_M\}$ are linearly dependent on $\K$ then they are linearly dependent on $\Cyl$, that is, just linearly dependent. Consequently, if the operators are linearly independent then they are linearly independent when restricted to $\K$.   

\paragraph{Step 2} Assume that $\{\hat{\varphi}_1,\ldots,\hat{\varphi}_M\}$ are linearly independent. Given $\kappa\in\K$, let us consider the following equation  
\[
\sum_ja_j\hat{\varphi}_j\kappa=0
\]    
imposed on unknown numbers $(a_1,\ldots,a_M)\in\R^M$. The set of all solutions of this equation is a linear subspace of $\R^M$ which will be denoted by $\mathbb{P}_\kappa$.

The linear independence of  $\{\hat{\varphi}_1,\ldots,\hat{\varphi}_M\}$ and the result of Step 1 mean that  
\begin{equation}
\bigcap_{\kappa\in\K}\mathbb{P}_\kappa=0\in\R^M.
\label{bigcap}
\end{equation}

Let us fix an elementary d.o.f. $\bar{\kappa}_1\in{\cal K}$. If $\mathbb{P}_{\bar{\kappa}_1}= 0$ then the operators $\{\hat{\varphi}_1,\ldots,\hat{\varphi}_M\}$ restricted to $\{\bar{\kappa}_1\}$ are linearly independent and we can proceed to Step 3. If $\mathbb{P}_{\bar{\kappa}_1}\neq 0$ then suppose  that for every $\kappa\in{\cal K}\setminus\{\bar{\kappa}_1\}$ 
\[
\mathbb{P}_{\kappa}=\mathbb{P}_{\bar{\kappa}_1}.
\]    
But this cannot be true because of \eqref{bigcap}. Hence there exists $\bar{\kappa}_2\in\K$ such that $\mathbb{P}_{\bar{\kappa}_2}\neq\mathbb{P}_{\bar{\kappa}_1}$ and consequently 
\[
\mathbb{P}_{\bar{\kappa}_1}\cap\mathbb{P}_{\bar{\kappa}_2}
\]
 is a linear subspace of $\R^M$ of the dimension {\em lower} than $\dim \mathbb{P}_{\bar{\kappa}_1}$.

If  $\mathbb{P}_{\bar{\kappa}_1}\cap\mathbb{P}_{\bar{\kappa}_2}=0$ then the operators $\{\hat{\varphi}_1,\ldots,\hat{\varphi}_M\}$ restricted to $\{\bar{\kappa}_1,\bar{\kappa}_{2}\}$ are linearly independent and we can proceed to Step 3. If $\mathbb{P}_{\bar{\kappa}_1}\cap\mathbb{P}_{\bar{\kappa}_2}\neq 0$ then suppose that for every $\kappa\in{\cal K}\setminus\{\bar{\kappa}_1,\bar{\kappa}_2\}$
\[
\mathbb{P}_\kappa=\mathbb{P}_{\bar{\kappa}_1}\cap\mathbb{P}_{\bar{\kappa}_2}.
\]     
But again this cannot be true because of \eqref{bigcap}. Hence there exists $\bar{\kappa}_3\in\K$ such that $\mathbb{P}_{\bar{\kappa}_3}\neq\mathbb{P}_{\bar{\kappa}_1}\cap\mathbb{P}_{\bar{\kappa}_2}$ and consequently
\[
\mathbb{P}_{\bar{\kappa}_1}\cap\mathbb{P}_{\bar{\kappa}_2}\cap\mathbb{P}_{\bar{\kappa}_3}
\]
 is a linear subspace of $\R^M$ of the dimension {\em lower} than $\dim (\mathbb{P}_{\bar{\kappa}_1}\cap\mathbb{P}_{\bar{\kappa}_2})$.

It is clear now that after a finite number of such steps we obtain a set $\{\bar{\kappa}_1,\ldots,\bar{\kappa}_{\bar{M}}\}$, ($\bar{M}\leq M$) such that
\[
\mathbb{P}_{\bar{\kappa}_1}\cap\ldots\cap\mathbb{P}_{\bar{\kappa}_{\bar{M}}}=0\in\R^M
\] 
which means that the operators $\{\hat{\varphi}_1,\ldots,\hat{\varphi}_M\}$ restricted to $\{\bar{\kappa}_1,\ldots,\bar{\kappa}_{\bar{M}}\}$ are linearly independent.
 
\paragraph{Step 3} By virtue of Assumption \ref{k-Lambda} there exists $(\hat{F},K)\in \Lambda$  such that each d.o.f. in $\{\bar{\kappa}_1,\ldots,\bar{\kappa}_{\bar{M}}\}$  is a cylindrical function compatible with $K$. The result of Step 2 together with Lemma \ref{cyl-K} mean that the operators $\{\hat{\varphi}_1,\ldots,\hat{\varphi}_M\}$ are linearly independent on $K$. 
\end{proof}

%***************************************************
\section{By-product: a Hilbert space built from almost periodic functions \label{by-prod}}
%***************************************************

It turns out that in the case of each theory for which it is possible to construct the space $\D$ of quantum states as described above it is also possible to construct a Hilbert space using only the configurational elementary d.o.f.. The construction we are going to present below applies almost periodic functions which some time ago were used to build a Hilbert space for Loop Quantum Cosmology---see \cite{mlqc,lqc} and references therein. 

Consider a collection $\K=\{\kappa\}$ of configurational elementary d.o.f. and suppose that there exists a directed set $(\mathbb{K},\geq)$ such that  
\begin{enumerate}
\item each element $K$ of $\mathbb{K}$ consists of a finite number of independent d.o.f. belonging to $\K$;  
\item for each $K\in\mathbb{K}$ the image of the map $\tilde{K}$ is $\R^N$, where $N$ is the number of elements of $K$;
\item if $K'\geq K$ then each elementary d.o.f in $K$ is a linear combination of d.o.f. in $K'$;   
\item if $Q_{K'}=Q_{K}$ then $K'\geq K$;  
\item for every finite set $K_0$ of configurational elementary d.o.f. there exists $K\in\mathbb{K}$ such that each $\kappa\in K_0$ is a cylindrical function compatible with $K$.
\end{enumerate} 

It turns out that every directed set $(\Lambda,\geq)$ satisfying Assumptions listed in Section \ref{ad-as} provides us with a directed set $(\mathbb{K},\geq)$ satisfying the requirements listed above: the set $\mathbb{K}$ is a subset of $\mathbf{K}$ such that for every $K\in\mathbb{K}$ there exists $\hat{F}\in\hat{\mathbf{F}}$ such that $(\hat{F},K)\in\Lambda$. Given $K',K\in\mathbb{K}$, we say that $K'\geq K$ if there exists $\hat{F}',\hat{F}\in\hat{\mathbf{F}}$ such that $(\hat{F}',K'),(\hat{F},K)\in\Lambda$ and $(\hat{F}',K')\geq(\hat{F},K)$.                 

Let $Q^*_K$ be the dual vector space to $Q_K$. Given $b\in Q^*_K$, we define a function
\begin{equation}
Q_K\ni a\mapsto e_b(a):=\exp(ib(a))\in\C.
\label{eb}
\end{equation}
Each finite linear combination
\[
\sum_i \alpha_ie_{b_i}
\]
with $\alpha_i\in \C$ is an {\em almost periodic function} on $Q_K$. On the space of all almost periodic functions on $Q_K$ we define a scalar product requiring that
\[
\scal{e_b}{e_{b'}}=
\begin{cases}
1 & \text{if $b=b'$ }\\
0 & \text{otherwise}
\end{cases}.
\] 
The space of the almost periodic functions may be completed with respect to the norm defined by the scalar product---the result is a (non-separable) Hilbert space which will be denoted by $\h_K$. 
 
Consider now $K'\geq K$ and a function \eqref{eb} on $Q_K$. Then by virtue of Lemma \ref{lm-pr} there exists a linear projection $\pr_{KK'}:Q_{K'}\to Q_K$ and  the pull-back
\[
\pr_{KK'}^*\,e_b=e_{\,b\,\circ\,\pr_{KK'}}
\]
is a function on $Q_{K'}$ of the sort \eqref{eb}---note that because of the linearity of $\pr_{KK'}$ the composition $b\circ\pr_{KK'}$ is an element of $Q^*_{K'}$. This means that the pull-back $\pr^*_{KK'}$ maps the orthonormal basis $\{e_b\}_{b\in Q^*_K}$ of $\h_K$ onto an orthonormal system of vectors in $\h_{K'}$. Therefore the map $U_{K'K}:\h_K\to \h_{K'}$ defined as the continuous extention of $\pr^*_{KK'}$ to the whole $\h_K$ is invertible and preserves the scalar products. 

It follows from \eqref{pr-KKK} that for each triplet $K''\geq K'\geq K$
\[
U_{K''K}=U_{K''K'}\circ U_{K'K}.
\]       
Thus $\{\h_K,U_{K'K}\}_{K\in\mathbb{K}}$ is an {\em inductive family} of Hilbert spaces and its inductive limit
\[
\h:=\underrightarrow{\lim} \,\h_K
\]   
is naturally a Hilbert space.

{So far we are not able to suggest any possible application of the Hilbert space just constructed to physical problems (see however a discussion at the end of Section \ref{D-quant}) but since it seems to be a natural by-product of the construction of the space $\D$ we decided to describe it briefly hoping that perhaps in the future the Hilbert space will turn out to be useful.}

%***************************************************
\section{An example---the degenerate Pleba\'nski gravity \label{DPG}}
%***************************************************

As an example we will build a space of quantum states for DPG introduced in \cite{oko-ncomp}\footnote{Let us note that in \cite{oko-ncomp} we constructed a space of quantum states for DPG, but that construction is essentially different from one we are going to describe below---the construction in \cite{oko-ncomp} is based on a directed set built without any reference to the momentum space $P$ of DPG and thus without any reference to the Poisson structure on the phase space of DPG; moreover it uses configurational elementary d.o.f. which are invariant with respect to transformations generated by a constraint of DPG.}. The Hamiltonian configuration space $Q$ of this theory is an affine space which, in a sense, is an ``almost linear'' space and therefore this example may cause a feeling of insufficiency since the general construction presented in this paper was invented to deal with non-linear Hamiltonian configuration spaces. However, carrying out the construction for DPG we {\em do not} take advantage of the affine structure of $Q$ to linearize this space. Note also that the starting point of the general construction is the presupposition that there exists a directed set $(\Lambda,\geq)$ satisfying Assumptions listed in Section \ref{ad-as}. Despite the simplicity of the space $Q$ the example illustrates quite well in which way such a directed set may be constructed; in particular, it proves usefulness of the auxiliary facts presented in Section \ref{aux}. Most importantly, methods used to construct the directed set $(\Lambda,\geq)$ for DPG when appropriately modified can be applied to build such a set for TEGR \cite{q-suit,ham-nv,q-stat} being a theory of genuinely non-linear Hamiltonian configuration space. Thus the reader may treat the construction described below as a preparatory exercise for studying the analogous but more complicated construction for TEGR presentation of which is too long to be included in this paper.  

%***************************************************
\subsection{Outline of the theory}
%***************************************************
 
DPG is a background independent theory of three fields defined on a four-dimensional manifold $\cal M$:
\begin{enumerate}
\item a two-form $\bld{\sigma}$; 
\item a one-form $\mathbf{A}$ which represents a connection on a trivial principal bundle ${\cal M}\times\R$, where the set $(\R,+)$ of real numbers with the addition is treated as a Lie group, that is, as the structure group of the bundle;    
\item a zero-form (a function) $\bld{\Theta}$.  
\end{enumerate}
The dynamics of the theory is given by the following action
\begin{equation}
S[\bld{\sigma},\mathbf{A},\bld{\Theta}]:=\int_{\cal M} \bld{\sigma}\wedge d\mathbf{A} -\frac{1}{2}\bld{\Theta}\, \bld{\sigma}\wedge\bld{\sigma},
\label{action}
\end{equation} 
where the exterior derivative $d\mathbf{A}$ is the curvature of the connection $\mathbf{A}$.

As described in \cite{oko-ncomp} this action is a simplification of the Pleba{\'n}ski self-dual action \cite{pleb} of general relativity and the theory under consideration coincides with so called $1+1$ degenerate sector of general relativity described in \cite{jac}. This fact justifies the name of the theory.

Assuming that ${\cal M}=\R\times\Sigma$, where $\Sigma$ is a compact three-dimensional manifold without boundary and treating $\R$ as a ``time axis''  and $\Sigma$ as a ``space'' we obtain the following Hamiltonian formulation of the theory: the space $P$ of momenta consists of all two-forms on $\Sigma$ while the Hamiltonian configuration space $Q$ coincides with the space of all connections on the trivial principal bundle $\Sigma\times\R$. The Hamiltonian reads as follows \cite{oko-ncomp}:
\[
H[\sigma, A, \alpha,\vec{N} ]=-\int_{\Sigma} \alpha\, d\sigma+\sigma\wedge(\vec{N}\lrcorner \,dA),
\]      
where $\sigma$ is a two-form playing the role of the momentum conjugate to the connection one-form $A$, $dA$ is the curvature of $A$, $\alpha$ is a function on $\Sigma$ and $\vec{N}$ a vector field on the manifold. The latter two variables are Lagrange multipliers.

Now let us proceed to the construction of a space of quantum states for DPG.

%***************************************************
\subsection{The construction}
%***************************************************

Let us emphasize that a construction of quantum states for the theory under considerations amounts to a construction of a directed set $(\Lambda,\geq)$ satisfying all Assumptions listed in Section \ref{ad-as} since, as already shown, each such a set unambiguously defines a space of quantum states. 

Because DPG is a background independent theory we will carry out the construction of the set $(\Lambda,\geq)$ in a background independent manner.   

%***************************************************
\subsubsection{Submanifolds of $\Sigma$}
%***************************************************

We start the construction by choosing elementary d.o.f.. Since we wish the construction to be background independent our choice of d.o.f. will be motivated by LQG methods (see e.g. \cite{acz,cq-diff,rev,rev-1} and references therein), where some essential functions on the phase space are associated with submanifolds of a space-like slice of a space-time. Thus each elementary d.o.f. will be associated with a submanifold of $\Sigma$. 

Assume that the manifold $\Sigma$ is {\em real analytic} and {\em oriented}.

An {\em analytic edge} is a one-dimensional connected analytic embedded submanifold of $\Sigma$ with two-point boundary. An {\em oriented} one-dimensional connected $C^0$ submanifold of $\Sigma$ given by a finite union of analytic edges will be called an {\em edge}. The set of all edges in $\Sigma$ will be denoted by $\cal E$.  

Given an edge $e$ of two-point boundary, its orientation  allows to call one of its endpoints {\em a source} and the other {\em a target} of the edge; if an edge is a loop then we distinguish one of its points and treat it simultaneously as the source and the target of the edge.

An edge $e^{-1}$ is called an {\em inverse} of an edge $e$ if $e^{-1}$ and $e$ coincide as un-oriented submanifolds of $\Sigma$ and differ by orientations.  We say that an edge $e$ is a composition of the edges $e_1$ and $e_2$, $e=e_2\circ e_1$, if $(i)$ $e$ as an oriented manifold is a union of $e_1$ and $e_2$, $(ii)$ the target of $e_1$ coincides with the source of $e_2$ and $(iii)$ $e_1\cap e_2$ consists solely of some (or all) endpoints of $e_1$ and $e_2$.               

We say that two edges are {\em independent} if the set of their common points is either empty or consists solely of some (or all) endpoints of the edges. A {\em graph} in $\Sigma$ is a finite set of pairwise independent edges. Any finite set of edges can be described in terms of edges of a graph \cite{al-hoop}: 
\begin{lm}
For every finite set $E=\{e_1,\ldots,e_N\}$ of edges there exists a graph $\gamma$ in $\Sigma$ such that every $e_i\in E$ is a composition of some edges of $\gamma$ and the inverses of some edges of the graph. The graph $\gamma$ can be chosen in such a way that
\[
\bigcup_{i=1}^Ne_i=\bigcup_{j=1}^Me'_j,
\]  
where the edges $\{e'_1,\ldots,e'_M\}$ constitute the graph $\gamma$. 
\label{E-gamma}
\end{lm}

The set of all graphs in $\Sigma$ is naturally a directed set: $\gamma'\geq\gamma$ if each edge of the graph $\gamma$ is a composition of some edges of the graph $\gamma'$ and the inverses of some edges of $\gamma'$. 
 
Let $S$ be a two-dimensional embedded submanifold of $\Sigma$. Assume that $S$ is $(i)$ analytic, $(ii)$ oriented and $(iii)$ of a compact closure. We moreover require $S$ to be such that every edge $e$ can be {\em adapted} to $S$ in the following sense \cite{area}: every edge $e$ can be divided into a finite number of edges $\{e_1,\ldots,e_N\}$, i.e.
\[
e=e_N\circ e_{N-1}\circ\ldots\circ e_2\circ e_1,
\] 
each of them either
\begin{enumerate}
\item is contained in the closure $\overline{S}$; 
\item has no common points with $S$; 
\item has exactly one common point with $S$ being one of its two distinct endpoints.    
\end{enumerate}
We will call such a submanifold a {\em face}. A set of all faces in $\Sigma$ will be denoted by $\cal S$.

%***************************************************
\subsubsection{Elementary degrees of freedom}
%***************************************************

Every edge $e\in{\cal E}$ defines naturally a function 
\begin{equation}
Q\ni A\mapsto \kappa_e(A):=\int_eA\in\R.
\label{k-e}
\end{equation}
We choose the set $\K$ of configurational elementary d.o.f.  to be a collection of functions \eqref{k-e} given by all edges in $\Sigma$:
\[
\K:=\{\ \kappa_e \ | \ e\in{\cal E}\ \}.
\]     
Functions in $\K$ separate points in $Q$.

Note that for every edge $e\in{\cal E}$ and every pair $e_1,e_2\in{\cal E}$ such that the composition $e_2\circ e_1$ is well defined   
\begin{align}
&\kappa_{e^{-1}}=-\kappa_{e}, && \kappa_{e_2\circ e_1}=\kappa_{e_2}+\kappa_{e_1}.
\label{keke}
\end{align}

Every face $S\in{\cal S}$ defines naturally a function 
\begin{equation}
P\ni \sigma\mapsto \varphi_S(\sigma):=\int_S\sigma\in\R.
\label{phi-S}
\end{equation}      
We choose the set $\F$ of momentum elementary d.o.f  to be a collection of functions \eqref{phi-S} given by all faces in $\Sigma$:
\[
\F:=\{\ \varphi_S\ | \ S\in{\cal S}\ \}.
\]     
Functions in $\F$ separate points in $P$.

%***************************************************
\subsubsection{Finite sets of configurational elementary d.o.f.}
%***************************************************

Every graph $\gamma$ defines a set $K_\gamma$ of configurational elementary d.o.f.---if $\{e_1,\ldots,e_N\}$ are edges constituting the graph then 
\[
K_\gamma:=\{\kappa_{e_1},\ldots,\kappa_{e_N}\}.
\] 
\begin{lm}
Let $\gamma=\{e_1,\ldots,e_N\}$ be a graph in $\Sigma$. Then for every $(x_1,\ldots,x_N)\in\R^N$ there exists $A\in Q$ such that
\[
\kappa_{e_i}(A)=x_i.
\]    
\label{QKg-RN-lm}
\end{lm}
\noindent The lemma being a simple modification of a lemma proven in \cite{al-hoop} means that for every graph $\gamma$ the image of $\tilde{K}_\gamma$ given by \eqref{k-inj} is $\R^N$, where $N$ is the number of edges of $\gamma$. In other words,
\begin{equation}
Q_{K_\gamma}\cong\R^N.
\label{QKg-RN}
\end{equation}
Consequently,  each $K_{\gamma}$ consists of independent d.o.f. and  $Q_{K_\gamma}$ is a reduced configuration space. Note that each one-element subset $K=\{\kappa_e\}$ of $\K$ is of the form $K_{\gamma}$ with $\gamma=\{e\}$.   
\begin{lm}
Let $K$ be a finite set of configurational elementary d.o.f.. Then there exists a graph $\gamma$ such that each d.o.f. in $K$ is a linear combination of d.o.f. in $K_{\gamma}$.
\label{f-Kg}
\end{lm}
\begin{proof}
The lemma follows immediately from Lemma \ref{E-gamma} and Equations \eqref{keke}.
\end{proof}

Now we are going to state and prove a lemma which will be relevant not only for the construction of quantum states for DPG but also for TEGR \cite{q-stat}. Therefore the lemma will be formulated in a quite general manner.

\begin{lm}
Let $\Omega$ be a set of one-forms on $\Sigma$ such that for every graph $\gamma_0=\{e_1,\ldots,$ $e_{N_0}\}$ and for each $(x_1,\ldots,x_{N_0})\in\R^{N_0}$ there exists $\varpi\in\Omega$ such that
\[
\kappa_{e_i}(\varpi)=x_i.
\]    
Then $\gamma'\geq\gamma$ if and only if each d.o.f. in $K_\gamma$ restricted to $\Omega$  is a linear combination of d.o.f. in $K_{\gamma'}$ restricted to $\Omega$.
\label{g'g-lin-gen}  
\end{lm}

Before we will prove the lemma let us emphasize that the configurational elementary d.o.f. are defined on $Q$ which can be naturally identified with the set of all one-forms on $\Sigma$. Therefore the d.o.f. are well defined on the set $\Omega$ introduced in the lemma.   

\begin{proof}[Proof of Lemma \ref{g'g-lin-gen}]
Since now till the end of the proof we will use exclusively the configurational d.o.f. {\em restricted to} $\Omega$ without mentioning this and without any change of the notation.

Suppose that $\gamma'\geq\gamma$. It is enough to apply Equations \eqref{keke} to the definition of the relation $\geq$ to conclude that each d.o.f. in $K_\gamma$ is a linear combination of $K_{\gamma'}$.

Assume now that each d.o.f. in $K_\gamma$ is a linear combination of $K_{\gamma'}$, 
\begin{equation}
\kappa_{e_i}=B_i{}^j\kappa_{e'_j},
\label{keBke-0}
\end{equation}
where $\gamma=\{e_1,\ldots,e_N\}$, $\gamma'=\{e'_1,\ldots,e'_{N'}\}$ and $\{B_i{}^j\}$ are real numbers.

Let us fix an edge $e_i\in\gamma$. Now we will prove the following two statements (note that both ignore orientations of edges):
\begin{enumerate}
\item $e_i$ is contained in the union of all edges of $\gamma'$;  
\item every edge  $e'_j$ of $\gamma'$ either is contained in $e_i$ or the set of their common points is finite.
\end{enumerate}

To prove Statement 1 assume that $e_i$  is a composition of some edges $\bar{e}_{1},\bar{e}_2$ and $\bar{e}_3$:
\[
e_i=\bar{e}_3\circ \bar{e}_2\circ \bar{e}_{1}
\]
such that
\[
\bar{e}_2\cap(e'_1\cup\ldots\cup e'_{N'})=\varnothing.
\]
Lemma \ref{E-gamma} ensures us that there exists a graph $\gamma''$ such that each $(i)$ edge in $\{\bar{e}_{1},\bar{e}_3\}\cup\gamma'$ is a combination of edges in $\gamma''$ and their inverses and $(ii)$ $\{\bar{e}_2\}\cup\gamma''$ is a graph. There exists $\varpi\in \Omega$ such that $\kappa_{\bar{e}_2}(\varpi)=1$ and $\kappa_{e''}(\varpi)=0$ for every edge $e''$ of $\gamma''$. This implies that $\kappa_{e_i}(\varpi)=1$ and $\kappa_{e'_j}(\varpi)=0$ for every edge $e'_j$ of $\gamma'$. But this is in contradiction with \eqref{keBke-0}, hence
\[
e_i\subset(e'_1\cup\ldots\cup e'_{N'}).
\]   

To prove Statement 2 let us fix $e'_j\in\gamma'$ and assume that it is a composition of some edges $\bar{e}_{1},\bar{e}_2$ and $\bar{e}_3$ 
\[
e'_j=\bar{e}_3\circ \bar{e}_2\circ \bar{e}_{1}
\]
such that
\[
\bar{e}_2\cap e_i =\varnothing.
\]
Let $\gamma''$ be a graph such that $(i)$ each edge in $\{e_i,\bar{e}_{1},\bar{e}_3\}\cup(\gamma'\setminus\{e'_j\})$ is a combination of edges in $\gamma''$ and their inverses and $(ii)$ $\{\bar{e}_2\}\cup\gamma''$ is a graph. There exists $\varpi\in \Omega$ such that $\kappa_{\bar{e}_2}(\varpi)=1$ and $\kappa_{e''}(\varpi)=0$ for every edge $e''$ of $\gamma''$. This means that $\kappa_{e_i}(\varpi)=0$, $\kappa_{e'_j}(\varpi)=1$ and $\kappa_{e'_l}(\varpi)=0$ for every $e'_l\in \gamma'\setminus\{e'_j\}$. Setting this $\varpi$ to \eqref{keBke-0} we obtain       
\begin{equation}
0=B_i{}^j.
\label{b-0}
\end{equation}

Suppose now that the same $e'_j$ is a composition of some edges $\bar{e}_{1},\bar{e}_2$ and $\bar{e}_3$ such that $\bar{e}_2\subset e_i$. Then $e_i$ is a composition of $\bar{e}_2$ and some other edges. Let $E$ be a set consisting of $\bar{e}_{1}$, $\bar{e}_3$, the edges of $\gamma'$ except $e'_j$ and the edges composing $e_i$ except $\bar{e}_2$. Let $\gamma''$ be a graph such that $(ii)$ each edge in $E$ is a combination of edges in $\gamma''$ and their inverses and $(ii)$  $\{\bar{e}_2\}\cup\gamma''$ is a graph. There exists $\varpi\in \Omega$ such that $\kappa_{\bar{e}_2}(\varpi)=1$ and $\kappa_{e''}(\varpi)=0$ for every edge $e''$ of $\gamma''$. Consequently, $\kappa_{e'_j}(\varpi)=1$, $\kappa_{e_i}(\varpi)=1$ (if the orientation of $\bar{e}_2$ coincides with that of $e_i$) or $\kappa_{e_i}(\varpi)=-1$ (otherwise) and $\kappa_{e'_l}(\varpi)=0$ for every $e'_l\in \gamma'\setminus\{e'_j\}$. Setting this $\varpi$ to \eqref{keBke-0} we obtain       
\begin{equation}
\pm1=B_i{}^j.
\label{b-1}
\end{equation}

Suppose finally that the same $e'_j$ is a composition of edges such that one of them is contained in $e_i$ and an other one has no common points with $e_i$. But this cannot be true since then Equations \eqref{b-0} and \eqref{b-1} would have to hold simultaneously. This means that either $e'_j$ is contained in $e_i$ or $e_i\cap e'_j$ does not contain any edge---in the latter case $e_i\cap e'_j$ consists of finite number of points because both $e_i$ and $e'_j$ are unions of {\em analytic} edges \cite{al-hoop}.      
   
Thus we justified both Statements. They mean that  $e_i$ as an un-oriented submanifold of $\Sigma$ is a union of some edges of $\gamma'$ regarded as un-oriented submanifolds. Taking into account the orientations of the edges under consideration we see that $e_i$ is a composition of some edges of $\gamma'$ and the inverses of some edges of the graph. Hence $\gamma'\geq\gamma$.      
\end{proof}

\begin{cor}
$\gamma'\geq\gamma$ if and only if each d.o.f. in $K_\gamma$ is a linear combination of d.o.f. in $K_{\gamma'}$.
\label{g'g-lin}  
\end{cor}
\begin{proof}
The configurational d.o.f. are defined on $Q$. This fact and Lemma \ref{QKg-RN-lm} allow to set $\Omega=Q$ in Lemma \ref{g'g-lin-gen}.
\end{proof}

An important observation is that by virtue of Lemmas \ref{QKg-RN-lm} and \ref{f-Kg} a subset of $\mathbf{K}$ consisting of all sets $K_{\gamma}$, where $\gamma$ runs through all graphs in $\Sigma$, meets the condition imposed on the set $\mathbf{K'}$ by Proposition \ref{big-pro}. Thus Assertions 1 of the proposition holds for all reduced configuration spaces which means in particular that the space $\Cyl$ given by the configurational elementary d.o.f. \eqref{k-e} is well defined and Assertion 2 holds for cylindrical functions compatible with elements of $\mathbf{K}$. Moreover, according to Assertion 3 for every $\Psi\in\Cyl$ there exists a graph $\gamma$ such that $\Psi$ is a cylindrical function compatible with $K_{\gamma}$.

%***************************************************
\subsubsection{Operators corresponding to momentum d.o.f.}
%***************************************************

As shown in \cite{acz} each elementary d.o.f. $\varphi_S\in\F$ defines via a suitably regularized Poisson bracket a linear operator $\hat{\varphi}_S$ acting on cylindrical functions. We concluded a while ago that every $\Psi\in\Cyl$ is of the form $\Psi=\pr_{K_\gamma}^*\psi$ for a graph $\gamma=\{e_1,\ldots,e_N\}$  and a smooth complex function $\psi$ on $Q_{K_\gamma}$.  The operators $\hat{\varphi}_S$ acts on $\Psi$ as follows:   
\begin{equation}
\hat{\varphi}_S\Psi:=\sum_{i=1}^N\eps(S,e_i)\, \pr_{K_\gamma}^*\partial_{x_i}\psi,
\label{hphi_S}
\end{equation}
where $(i)$ $\{\partial_{x_i}\}$ are vector fields on $Q_{K_\gamma}$ given by a coordinate frame $(x_i)$ defined by the d.o.f. in $K_\gamma$ according to \eqref{lin-coor} and $(ii)$ $\eps(S,e_i)$ is a real number associated with the edge $e_i$ according to the following prescription.      

Adapting the edge $e_i$ to $S$ we obtain a set of edges $\{e_{i1},\ldots,e_{in}\}$ and define a function $\eps$ on this set: $\eps(e_{ia})=0$ if $e_{ia}$ is contained in $\bar{S}$ or has no common points with $S$; otherwise           
\begin{enumerate}
\item $\eps(e_{ia}):=\frac{1}{2}$ if $e_{ia}$ is either 'outgoing' from $S$ and placed 'below' the face  or is 'incoming' to $S$ and placed 'above' the face;
\item $\eps(e_{ia}):=-\frac{1}{2}$ if $e_{ia}$ is either 'outgoing' from $S$ and placed 'above' the face  or is 'incoming' to $S$ and placed 'below' the face. 
\end{enumerate}
Here the terms 'outgoing' and 'ingoing' refer to the orientation of the edges (which is inherited from the orientation of $e_i$) while the terms 'below' and 'above' refer to the orientation of the normal bundle of $S$ defined naturally by the orientations of $S$ and $\Sigma$. Then we define
\[
\eps(S,e_i):=\sum_{a=1}^n \eps(e_{ia}).
\]       

Note finally that for every $S\in{\cal S}$ and every $e\in{\cal E}$  
\begin{equation}
\hat{\varphi}_S\kappa_{e}=\eps(S,e)
\label{hphiS-ke}
\end{equation}
which means that $\hat{\varphi}_S\kappa_{e}$ is a real {\em constant} cylindrical function. 

Operators $\{\ \hat{\varphi} \ | \ \varphi\in\F\ \}$ span the real linear space $\hat{\F}$. Any element $\hat{\varphi}\in \hat{\F}$ can be expressed as
\[
\hat{\varphi}=\sum_i \alpha_i\hat{\varphi}_{S_i},
\]  
where $\{\alpha_i\}$ are real numbers and the sum is finite. Suppose that $\Psi=\pr_{K_\gamma}\psi$ is a cylindrical function compatible with $K_\gamma$. By virtue of  Equations \eqref{hphi_S} and \eqref{hphiS-ke} the action of the operator $\hat{\varphi}$ on $\Psi$ reads as follows:
\begin{equation}
\hat{\varphi}\Psi=\sum_{ij}\alpha_i\eps(S_i,e_j)\, \pr_{K_\gamma}^*\partial_{x_j}\psi=\sum_i\Big(\pr_{K_\gamma}^*\partial_{x_i}\psi\Big)\,\hat{\varphi}\kappa_{e_i}.
\label{hphi-Psi}
\end{equation}

%***************************************************
\subsubsection{A directed set $\Lambda$}
%***************************************************

Now we are going to choose a directed set $(\Lambda,\geq)$. Let $\hat{F}$ be an element of $\hat{\mathbf{F}}$ and $K=\{\kappa_1,\ldots,\kappa_{N}\}$ an element of $\mathbf{K}$. A pair $(\hat{F},K)$ is said to be {\em non-degenerate} if $\dim\hat{F}=N$ and an $(N\times N)$-matrix $G=(G_{ji})$ of components          
\begin{equation}
G_{ji}:=\hat{\varphi}_{j}\kappa_{i},
\label{matr-G}
\end{equation}
where $(\hat{\varphi}_1,\ldots,\hat{\varphi}_{N})$ is a basis of $\hat{F}$, is non-degenerate.

\begin{df}
The set $\Lambda$ is a set of all non-degenerate pairs $(\hat{F},K_{{\gamma}})\in\hat{\mathbf{F}}\times \mathbf{K}$, where ${\gamma}$ runs through all graphs in $\Sigma$.    
\label{df-Lambda}
\end{df}

\begin{lm} 
For every graph $\gamma$ in $\Sigma$ there exists $\hat{F}\in\hat{\mathbf{F}}$ such that $(\hat{F},K_\gamma)\in\Lambda$.
\label{every-g}
\end{lm}
\begin{proof}
The independence of edges $\{e_1,\ldots,e_N\}$ of the graph $\gamma$ imply that there exists a set $\{S_1,\ldots,S_N\}$ of faces such that $e_i\cap S_j$ is empty if $i\neq j$ and consists of exactly one point distinct from the endpoints of $e_i$ if $i=j$. The orientations of the faces can be chosen in such a way that
\[
G_{ji}=\hat{\varphi}_{S_j}\kappa_{e_i}=\delta_{ji}.
\]
Now it is enough to define $\hat{F}:={\rm span}\,\{\hat{\varphi}_{S_1},\ldots,\hat{\varphi}_{S_N}\}$.      
\end{proof}

Now we are at the point to define  a relation $\geq$ on $\Lambda$: 
\begin{df}
Let $(\hat{F}',K_{\gamma'}),(\hat{F},K_\gamma)\in\Lambda$. We say that $(\hat{F}',K_{\gamma'})\geq (\hat{F},K_\gamma)$ if 
\begin{align*}
&\hat{F}'\supset\hat{F} && \text{and} && \gamma'\geq\gamma.         
\end{align*}
\label{df-Lambda->}
\end{df}

\begin{lm}
$(\Lambda,\geq)$ is a directed set.
\label{Lambda-dir}
\end{lm}

\begin{proof}
 The transitivity of the relation $\geq$ is obvious. Thus it has to be proven only that for any two elements $\lambda',\lambda\in\Lambda$ there exists $\lambda''\in\Lambda$ such that $\lambda''\geq\lambda'$ and $\lambda''\geq\lambda$.

To prove this we are going to use Proposition \ref{Lambda-pr}. Therefore we have to show first that the set $\Lambda$ satisfy Assumptions \ref{k-Lambda} and \ref{comp-f}. The set meets Assumption \ref{k-Lambda} due to Lemmas \ref{f-Kg} and \ref{every-g}. On the other hand the formula \eqref{hphi-Psi} imply immediately that Assumption \ref{comp-f} is satisfied.

Let us fix $\lambda'=(\hat{F}',K_{\gamma'})$ and $\lambda=(\hat{F},K_\gamma)\in\Lambda$. We define $\hat{F}^0$ as a linear subspace of $\hat{\F}$ spanned by elements of $\hat{F}'$ and $\hat{F}$ and choose a basis $(\hat{\varphi}_1,\ldots,\hat{\varphi}_M)$ of $\hat{F}^0$. Proposition \ref{Lambda-pr} and the definition of $\Lambda$ guarantee that there exists a graph $\gamma^0$ such that the operators $(\hat{\varphi}_1,\ldots,\hat{\varphi}_M)$ remain  linearly independent when restricted to $K_{\gamma^0}$. Let $\gamma''$ be a graph such that  the number of its edges is greater than $\dim \hat{F}^0=M$ and $\gamma''\geq \gamma^0,\gamma',\gamma$. By virtue of Corollary \ref{g'g-lin} d.o.f. in $K_{\gamma^0}$ are cylindrical functions compatible with $K_{\gamma''}$ and, according to Lemma \ref{cyl-K}, the operators $(\hat{\varphi}_1,\ldots,\hat{\varphi}_M)$  are linearly independent on $K_{\gamma''}$.     

Consider now a matrix $G^0$ of components
\[
G^0_{ji}:=\hat{\varphi}_j\kappa_{e_i},
\]          
where $\{e_1,\ldots,e_N\}$ constitute the graph $\gamma''$.  Clearly, the matrix has $M$ rows and $N$ columns and because the operators $(\hat{\varphi}_1,\ldots,\hat{\varphi}_M)$ are linearly independent on $K_{\gamma''}$ its rank is equal $M<N$. Using the following operations $(i)$ multiplying a row by a non-zero number, $(ii)$ adding to a row a linear combination of other rows $(iii)$ reordering the rows and $(iv)$ reordering the columns we can transform the matrix $G^0$ to a matrix $G^1$ of the following form
\[
G^1=
\begin{pmatrix}
\mathbf{1} & G' 
\end{pmatrix},
\]             
where $\mathbf{1}$ is $M\times M$ unit matrix and $G'$ is a $M\times(N-M)$ matrix. Note that the first three operations used to transform $G^0$ to $G^1$  correspond to a transformation of the basis $(\hat{\varphi}_1,\ldots,\hat{\varphi}_M)$ to an other basis $(\hat{\varphi}'_1,\ldots,\hat{\varphi}'_M)$ of $\hat{F}^0$, while the fourth operation corresponds to renumbering the edges of $\gamma''$: $e_i\mapsto e'_i:=e_{\sigma(i)}$, where $\sigma$ is a permutation of the sequence $(1,\ldots,N)$. Thus  
\[
G^1_{ji}=\hat{\varphi}'_{j}\kappa_{e'_i}.
\]

Let $\{\hat{\varphi}^0_1,\ldots,\hat{\varphi}^0_{N}\}$ be operators constructed  with respect to $K_{{\gamma}''}$ exactly as in the proof of Lemma \ref{every-g}. Then
\[
\hat{\varphi}^0_{j}\kappa_{e'_i}=\delta_{ji}.
\]   
Thus if
\[
\Big(\hat{\varphi}''_1,\ldots,\hat{\varphi}''_{N}\Big):=\Big(\hat{\varphi}'_1,\ldots,\hat{\varphi}'_{M},\hat{\varphi}^0_{M+1},\ldots,\hat{\varphi}^0_{N}\Big)
\]
then an $N\times N$ matrix $G=(G_{ji})$ of components 
\[
G_{ji}:=\hat{\varphi}''_{j}\kappa_{e'_{i}}
\]
is of the following form
\[
G=
\begin{pmatrix}
\mathbf{1} & G'\\
\mathbf{0} & \mathbf{1}' 
\end{pmatrix},
\]  
where $\mathbf{0}$ is a zero $(N-M)\times M$ matrix, and $\mathbf{1}'$ is a unit $(N-M)\times(N-M)$ matrix. It means in particular that the operators $(\hat{\varphi}''_{1},\ldots,\hat{\varphi}''_N)$ are linearly independent.

To finish the proof we define 
\[
\hat{F}'':={\rm span} \, \{\hat{\varphi}''_{1},\ldots,\hat{\varphi}''_N\}
\]
and $\lambda'':=(\hat{F}'',K_{\gamma''})$.       
\end{proof}

%***************************************************
\subsubsection{Checking Assumptions}
%***************************************************
 
Now let us check whether the directed set $(\Lambda,\geq)$ satisfies all Assumptions listed in Section \ref{ad-as}. 

Proving Lemma \ref{Lambda-dir} we showed that $\Lambda$ satisfies Assumption \ref{k-Lambda}. Regarding Assumption \ref{f-Lambda} let $F_0=\{\varphi_{S_1},\ldots,\varphi_{S_N}\}$. For each face $S_i$ there exists an edge $e_i$ such that 
\[
\hat{\varphi}_{S_i}\kappa_{e_i}=1.
\]    
Let $\hat{F}_i:={\rm span}\{\hat{\varphi}_{S_i}\}$ and $\gamma_i:=\{e_i\}$. Thus $(\hat{F}_i,K_{\gamma_i})\in\Lambda$ for every $i=1,\ldots,N$. Since $\Lambda$ is a directed set there exists $(\hat{F},K_\gamma)\in \Lambda$ such that $(\hat{F},K_\gamma)\geq(\hat{F}_i,K_{\gamma_i})$ for every $i=1,\ldots,N$. Taking into account  Definition \ref{df-Lambda->} of the relation $\geq$ on $\Lambda$ we see that $\hat{F}$ contains all the operators $\{\hat{\varphi}_{S_1},\ldots,\hat{\varphi}_{S_N}\}$. Thus Assumption \ref{f-Lambda} is satisfied.          

Assumption \ref{RN} is satisfied by virtue of Lemma \ref{QKg-RN-lm}. Proving Lemma \ref{Lambda-dir} we showed that $\Lambda$ meets Assumption \ref{comp-f}. Assumption \ref{const} follows immediately from Equation \eqref{hphiS-ke}. Assumption \ref{non-deg} is satisfied by virtue of Definition \ref{df-Lambda} of $\Lambda$.

Regarding Assumption \ref{Q'=Q} we note first that, given $K_{\gamma},K_{\gamma'}$, there exists $K_{\gamma''}$ such that each d.o.f. in $K_{\gamma}\cup K_{\gamma'}$ is a linear combination of d.o.f. in $K_{\gamma''}$ (see Lemma \ref{f-Kg}). Suppose that $Q_{K_\gamma}=Q_{K_{\gamma'}}$. Then Equation \eqref{QKg-RN} applied to $K_{\gamma''}$ allows us to use Proposition \ref{KK'barK} and conclude that every d.o.f. in $K_\gamma$ is a linear combination of d.o.f. in $K_{\gamma'}$. Then, as stated by  Corollary \ref{g'g-lin}, $\gamma'\geq \gamma$ and Assumption \ref{Q'=Q} follows.  

By virtue of Definition \ref{df-Lambda->} of the relation $\geq$ on $\Lambda$ $K_{\gamma'}\geq K_\gamma$ if only $\gamma'\geq\gamma$  and Assumption \ref{lin-comb} follows from Corollary \ref{g'g-lin}. Assumption \ref{FF'} is satisfied  due to the same definition. 

We conclude that the directed set $(\Lambda,\geq)$ constructed for DPG satisfies all Assumptions which means that for this theory there exists the corresponding space $\D$ of quantum states. Let us emphasize that the directed set is built in a background independent manner hence the same can be said about the resulting space $\D$.     

%***************************************************
\section{Discussion}
%***************************************************

%***************************************************
\subsection{General remarks}
%***************************************************
 
The main result of this paper can be summarized as follows: if for a phase space of a field theory there exists a directed set $(\Lambda,\geq)$ defined as in Section \ref{outline} and satisfying all Assumptions listed in Section \ref{ad-as} then for the phase space there exists a corresponding space {$\D$} of kinematic quantum states. Note that this result reduces the task of constructing such a space of quantum states to the task of constructing an appropriate directed set---this fact was already used in the case of DPG described in the previous section and will be used also in  \cite{q-suit,ham-nv,q-stat} to construct a space of quantum states for TEGR. 

{Let us emphasize that, given phase space, the results of this paper {\em do not} guarantee either an existence of a space $\D$ or a uniqueness of it---for example, in the case of TEGR we managed to construct two essentially different spaces of this sort based on different sets of elementary d.o.f. \cite{q-suit,q-stat}.}

{It seems that one reason for the possibility of constructing many different spaces $\D$ for the same phase space is the independence of the construction of mathematical properties of the phase space as an infinite dimensional ``manifold'' of fields. Note that to construct $\D$ we first reduce the infinite dimensional  phase space $P\times Q$ to a ``finite dimensional phase space'' represented by an element $(\hat{F},K)\in\Lambda$---the set $K$ defines the reduced configuration space $Q_K$ while a basis of $\hat{F}$ corresponds to a finite number of momentum d.o.f.. This is exactly this first step i.e. the passage from the infinite dimensional space to the reduced ``finite dimensional one'' where most information about the mathematical structure of the original phase space is lost.\footnote{A similar situation occurs in the case of a space $\Abar$ of so called {\em general connections} known from Loop Quantum Gravity \cite{al-gen}. Given a space $\cal A$ of {\em smooth} connections on a principle bundle $P(\Sigma,G)$ over a manifold $\Sigma$ with a structure group $G$, one uses graphs in $\Sigma$ to reduce the infinite dimensional space $\cal A$ to finite dimensional spaces associated with the graphs. Next, one ``glues'' by means of {\em projective techniques} all the reduced spaces into $\Abar$. It turns out that ${\cal A}\subset\Abar$, but in $\Abar$ there are also many non-smooth ``connections'' including distributional ones. Moreover, $\Abar$  can be defined in an algebraic way without any reference to the smooth connections constituting $\cal A$ as well as to the topology of the bundle $P(\Sigma,G)$, which means that $\Abar$ is independent of some mathematical properties of $\cal A$ and the bundle.}} 

{Let us note that even the differential structures on reduced configuration spaces introduced in Section \ref{pre} are independent of any differential structure on $Q$ or $P\times Q$. To see this recall that to define the differential structure on $Q_K$ we imposed a requirement on the image of $\tilde{K}$---this requirement can be equivalently formulated as follows:  the image of the map
\begin{equation}
\tilde{K}\circ \pr_K:Q\to \R^N, 
\label{im-K}
\end{equation}
where $\pr_K$ is given by (2.4) and $N$ is the number of elements of $K$, is an $N$-dimensional submanifold of $\R^N$. Let us emphasize that we {\em do not} require the map \eqref{im-K} to be {\em differentiable} and the image of \eqref{im-K} can be found without refering to any differential structure on $Q$ or $P\times Q$ as it was done in the cases of DPG (see Section \ref{DPG}) and TEGR \cite{q-suit}.}   

{Let us also remark on so called ``non-compactness'' problem of LQG. The gauge group of the complex Ashtekar variables \cite{a-var-1,a-var-2} is $SL(2,\C)$. Trying to base on these variables a construction of a Hilbert space for LQG via the inductive techniques one immediately encounters obstacles \cite{oko-ncomp} caused by non-compactness of resulting reduced configuration spaces which are isomorphic to Cartesian products of finite numbers of $SL(2,\C)$. Since non-compactness of configurational spaces of a classical system and its subsystem is not an obstacle for defining a partial trace projecting states of a corresponding quantum system onto states of its corresponding subsystem one may hope to solve the non-compactness problem of LQG by means of projective techniques. However, the present construction does not seem to be very helpful for achieving this goal since it is highly dependent on linearity of configuration spaces and maps between them and any space $SL(2,\C)^N$ is obviously non-linear. Since finite dimensional linear spaces are non-compact Abelian Lie groups we can say that the inductive techniques applied currently in LQG work well in the case of reduced configuration spaces being compact Lie groups, while the projective techniques work well in the case of non-compact Abelian Lie groups. We leave open a question whether it is possible to generalize the projective techniques to the case of {\em non-compact non-Abelian} Lie groups.} 

%***************************************************
\subsection{The space $\D$ and quantization \label{D-quant}}
%***************************************************

Given theory, a space $\D$ of quantum states constructed for it according to the prescription presented in this paper is {{\em not a Hilbert space} but rather a convex set of states and is} usually too large to serve as a base of a quantum model of the theory. However, it is possible to construct a Hilbert space from the space $\D$ which may serve this purpose. The construction reads as follows \cite{kpt}.

Recall that, given element $\lambda$ of $\Lambda$, we denoted by ${\cal B}_\la$ the $C^*$-algebra of bounded linear operators on the Hilbert space $\h_\la$. The space $\D_\la$ of all density operators on $\h_\la$  can be treated as a space of all regular algebraic states on the  algebra ${\cal B}_\la$ i.e. linear $\C$-valued positive normed functionals  which acts on elements of ${\cal B}_\la$ via a trace: for every $\rho\in\D_\la$ the following map
\[
{\cal B}_\la\ni a \mapsto \tr(a\rho)\in\C
\]          
is a regular state. For every pair $\lambda'\geq\lambda$ of elements of $\Lambda$ there exists a unique injective $*$-homomorphism $\pi^*_{\lambda'\lambda}:{\cal B}_\la\to{\cal B}_{\la'}$ such that for every $a\in{\cal B}_\la$ and every $\rho'\in \D_{\la'}$
\[
\tr(\pi^*_{\la'\la}(a)\rho')=\tr(a\,\pi_{\la\la'}(\rho')),
\]   
where $\pi_{\la\la'}:\D_{\la'}\to\D_\la$ is the projection given by Definition \ref{forget}. It follows from Proposition \ref{pipipi} that for every triplet $\la''\geq\la'\geq\la$ 
\[
\pi^*_{\la''\la}=\pi^*_{\la''\la'}\circ\pi^*_{\la'\la},
\]  
which means that $\{{\cal B},\pi^*_{\la'\la}\}$ is an inductive family of $C^*$-algebras. Its inductive limit
\[
{\cal B}:= \underrightarrow{\lim}\,{\cal B}_\la
\]   
is naturally a $C^*$-algebra which can be interpreted as an algebra of quantum observables of the theory and each element $\rho$ of the space $\D$ defines a state on $\cal B$.    

Let us emphasize that the construction of both $\D$ and $\cal B$ is purely kinematic i.e. it uses merely properties of the phase space $P\times Q$ of a theory without any reference to its dynamics and possible constraints on the phase space. According to \cite{kpt} once both $\D$ and $\cal B$ are constructed one should refer to the dynamics---and to constraints if a theory is a constrained system---to single out an element of $\D$ and use it as a state on $\cal B$ to build via the GNS construction a Hilbert space for a quantum model of the theory. 

It is not clear, in general, which criterion should be used to single out such an element of $\D${, but if a classical field theory for which the space $\D$ is constructed possesses a vacuum state then a procedure described in \cite{kpt} may be applied. According to the procedure for every $\lambda\in\Lambda$ one defines a Hamiltonian $\hat{H}_\lambda$ as an operator on the Hilbert space $\h_\lambda$ corresponding to the Hamiltonian of the classical theory. Then one finds a vacuum state in $\h_\lambda$ defined by $\hat{H}_\lambda$. This vacuum pure state defines a unique state $\rho_\lambda\in\D_\lambda$. In the last step one organizes all the vacuum states $\{\rho_\lambda\}$ into an element of $\D$ (see \cite{kpt} for details) which by definition is the vacuum state for a quantum model of the theory. The vacuum state in $\D$ generates a Hilbert space as explained above and on this Hilbert space one may try to define a quantum dynamics and quantum constraints if necessary.}

{However, there are theories like general relativity where Hamiltonians are sums of constraints and in these cases the procedure mentioned above is not expected to work well---such theories may possess no vacuum state. In such a case one can modify the procedure and use it to define and solve quantum constraints directly on $\D$. Given constraint $C$ on the phase space of such a theory, one may define a family of operators $\{\hat{C}_\lambda\}_{\lambda\in\Lambda}$ corresponding to $C$ such that each $\hat{C}_\lambda$ acts on $\h_\lambda$ and may be interpreted as a quantum constraint on this Hilbert space. If for every $\lambda$ solutions of $\hat{C}_\lambda$ are elements of $\h_\lambda$ then one may use these pure solutions to obtain mixed solutions belonging to $\D_\lambda$. Then a solution in $\D$ of a quantum constraint corresponding to $C$ may be defined as e.g. an element $\rho=\{\rho_\lambda\}_{\lambda\in\Lambda}\in\D$ such that for every $\lambda$ the mixed state $\rho_\lambda\in\D_\lambda$ is a solution of $\hat{C}_\lambda$. Next, one may use the solutions in $\D$ to define Hilbert spaces via the GNS constructions. If it turned out that the Hilbert spaces do not coincide then they could be treated as distinct quantum sectors of the quantized theory.}

{It may also happen that the constraints $\{\hat{C}_\lambda\}$ do not have solutions in the corresponding Hilbert spaces. Then for every $\lambda$ one may search for solutions of $\hat{C}_\lambda$ in a space ${\cal S}^*_\lambda$ dual to a subspace ${\cal S}_\lambda$ of $\h_\lambda$---${\cal S}_\lambda$ can be chosen to be e.g. a space of Schwartz functions (smooth complex functions of rapid decrease) on the corresponding $Q_K$. In the next step one may try to organize solutions found in the spaces $\{{\cal S}^*_\lambda\}$ into a Hilbert space or a convex set of states. Note that if ${\cal S}_\lambda$ is the space of Schwartz functions on $Q_K$ then almost periodic functions on $Q_K$ considered  in Section \ref{by-prod} are elements of its dual space ${\cal S}^*_\lambda$ i.e. they are distributions on ${\cal S}_\lambda$. Thus the Hilbert space $\h$ constructed in Section \ref{by-prod} is built from some elements of the spaces $\{{\cal S}^*_\lambda\}$ and perhaps this construction may provide a hint how to organize solutions of this sort into a Hilbert space.}

{Obviously, the outline of the procedures above does not close the issue of defining and solving quantum constraints in a space $\D$ of kinematic quantum states---first of all it seems to be necessary to construct and study examples of quantum constraints on $\D$ and their solutions.}   

{Let us note finally that it is the dynamics (and constraints) of a theory which, hopefully, should allow to pick up not only a proper Hilbert space among those which can be constructed from a space $\D$ and a corresponding $C^*$-algebra $\cal B$, but also a proper space $\D$ among those which may be constructed from a phase space---e.g. if, given space $\D$, there are obstacles which make it impossible to define quantum constraints on it then one should abandon this space and try to construct an other one from the same phase space.}

%***************************************************
\subsection{The present construction versus the original one \label{disc-comp}}
%***************************************************

As already stated the present construction is a generalization of one introduced in \cite{kpt}. Here we are going to compare both constructions.

It is evident that both here and in \cite{kpt} linearity of spaces and maps applied in the constructions plays a crucial role. The main difference between the constructions is the source of the linearity. 

In \cite{kpt} it is {\em assumed} that the Hamiltonian configuration space $Q$ is a linear space. Since the momentum space $P$ is naturally linear  the whole phase space $P\times Q$ is equipped with a linear structure. Both momentum and configurational d.o.f. are linear functions on $P$ and $Q$ respectively and with every relation $\lambda'\geq\lambda$ there are associated appropriate linear maps between spaces constituting the elements $\lambda'$ and $\lambda$.

The present construction does not require the space $Q$ to be linear but needs instead Assumptions \ref{RN} and \ref{lin-comb} to be imposed on configurational d.o.f. in order to define linear structures on reduced configurations spaces $\{Q_K\}$ and to ensure linearity of projections $\{\pr_{KK'}\}$. Let us emphasize that imposing these two assumptions {\em does not} amount to linearizing the space $Q$, but rather to embedding it into a larger linear space. To see what this larger linear space is recall that it was shown in Section \ref{by-prod} that from every directed set $(\Lambda,\geq)$ satisfying Assumptions listed in Section \ref{ad-as} one can derive a directed set $(\mathbb{K},\geq)$. It is easy to realize that $\{Q_K,\pr_{KK'}\}_{K\in\mathbb{K}}$ is a projective family. By virtue of Assumptions \ref{RN} and \ref{lin-comb} the projective limit 
\[
\bar{Q}:=\underleftarrow{\lim} \,Q_K
\]
of the projective family is equipped with a linear structure and it is the larger linear space. There is a natural map from $Q$ into $\bar{Q}$
\[
Q\ni q\mapsto (\,[q]_K)_{K\in\mathbb{K}}\in \bar{Q},
\]                    
where $[q]_K$ denotes the equivalence class of $q$ given by the relation $\sim_K$. Since configurational elementary d.o.f. separate points in $Q$ and Assumption \ref{k-Lambda} holds the map above is {\em injective} and it is the embedding of $Q$ into the larger linear space $\bar{Q}$.       

Regarding momentum elementary d.o.f. of the present construction, since they are supposed to define via a Poisson bracket (or its regularization) an operator on the space of cylindrical functions they should be linear functions on $P$ \cite{acz}. Of course, configurational elementary d.o.f. cannot be said to be linear or not because $Q$ is not a linear space, but they cannot be arbitrary functions on $Q$ because of the linearity of momentum d.o.f and Assumption \ref{const}---assuming that $q\in Q$ is a tensor field on a manifold and configurational d.o.f. $\kappa$ is defined as an integral of a function constructed from components of $q$ then to satisfy Assumption \ref{const} the function should be a polynomial of the components of degree $1$.            

Another difference between the two construction concerns elementary d.o.f.: in \cite{kpt} the canonical variables are fields on a three-dimensional manifold and the d.o.f. are ``regular'' functionals of the fields i.e. defined via three-dimensional integrals, while in the present paper we do not assume any ``regularity'' of d.o.f.. In particular, we admit ``distributional'' functionals as ones used in the case of DPG where fields defined on a three-dimensional manifold were integrated along one-dimensional edges and over two-dimensional faces. However, Poisson brackets of such ``non-regular'' functionals  are not always well defined and this is why we followed \cite{acz} in applying operators $\{\hat{\varphi}\}$ defined via a regularization of a bracket. 

Let us note finally that in \cite{kpt} the projections $\{\pi_{\lambda\la'}\}$ are defined without introducing injections $\{\omega_{\la'\la}\}$ and therefore at the first sight it is not clear whether the prescription for constructing the projections used in \cite{kpt} coincides with one applied in the present paper. However, a closer examination shows that the prescription in \cite{kpt} is a particular case of the one presented here and the apparent difference comes only from the different descriptions.     

\paragraph{Acknowledgment} I am very grateful to Prof. Jerzy Kijowski for introducing to me his idea of constructing spaces of quantum states via projective techniques and for valuable discussions.

%***************************************************
%***************************************************

\end{document}